\documentclass{article}

\usepackage{arxiv}

\usepackage{amsmath}
\usepackage{amsthm}
\usepackage{amsfonts}
\usepackage{amssymb}

\usepackage[utf8]{inputenc} 
\usepackage[T1]{fontenc}    
\usepackage{esvect}
\usepackage{epstopdf}		
\usepackage{booktabs}       
\usepackage{nicefrac}       
\usepackage{graphicx}
\graphicspath{{./Figures/}}
\usepackage{natbib}

\usepackage{comment}
\usepackage{dsfont}
\usepackage[symbol]{footmisc}

\bibpunct[, ]{[}{]}{;}{a}{}{,}
\renewcommand\bibfont{\fontsize{10}{12}\selectfont}

\usepackage{algorithm}
\usepackage{algorithmicx}
\usepackage[noend]{algpseudocode}
\usepackage{multirow}
\usepackage[labelsep=space]{caption}
\usepackage{subcaption}

\theoremstyle{plain}
\newtheorem{theorem}{Theorem}[section]

\theoremstyle{definition}

\theoremstyle{remark}

\newcommand*\GRU{\textsc{GRU}}

\newcommand*\RNN{\textsc{RNN}}

\newcommand*\bm[1]{\mathbf{#1}}
\newcommand*\Let[2]{\State #1 $\gets$ #2}
\algrenewcommand\alglinenumber[1]{
    {\sf\footnotesize\addfontfeatures{Colour=888888,Numbers=Monospaced}#1}}
\algrenewcommand\algorithmicrequire{\textbf{input}}

\renewcommand{\thetable}{\arabic{table}}

\newcommand*\colrule{\\[-7.5pt]\hline\\[-5pt]}
\newcommand*\botrule{\\[-7.7pt]\hline}

\title{Efficient Pricing and Hedging of High Dimensional American Options using Deep Recurrent Networks}

\author{{Andrew S.~Na}\thanks{CONTACT Andrew S. Na. Email: andrew.na@uwaterloo.ca} \\
	David R. Cheriton School of Computer Science\\
	University of Waterloo\\
	Waterloo, ON\\
	\texttt{andrew.na@uwaterloo.ca} \\
	\And
	{Justin W.L.~Wan} \\
	David R. Cheriton School of Computer Science\\
	University of Waterloo\\
	Waterloo, ON\\
	\texttt{justin.wan@uwaterloo.ca}
}

\renewcommand{\undertitle}{To appear in Quantitative Finance and accepted for publication January 6, 2023}

\begin{document}
\maketitle

\begin{abstract}
We propose a deep Recurrent neural network (RNN) framework for computing prices and deltas of American options in high dimensions. Our proposed framework uses two deep RNNs, where one network learns the price and the other learns the delta of the option for each timestep. Our proposed framework yields prices and deltas for the entire spacetime, not only at a given point (e.g. t = 0). The computational cost of the proposed approach is linear in time, which improves on the quadratic time seen for feedforward networks that price American options. The computational memory cost of our method is constant in memory, which is an improvement over the linear memory costs seen in feedforward networks. Our numerical simulations demonstrate these contributions, and show that the proposed deep RNN framework is computationally more efficient than traditional feedforward neural network frameworks in time and memory.
\end{abstract}


\section{Introduction}

Neural network solutions to the American option problem are some of the most studied topics in financial applications of neural networks. This is because, the capabilities of neural networks to overcome the curse of dimensionality is attractive to practitioners who price American options products. American options are one of the most commonly traded option styles in the  options market and the pricing and hedging of American options presents interesting academic challenges. However, many papers that look at neural networks in option pricing do not account for the extended training time required to train a network for each timestep to achieve good price estimates. Financial institutions that quote and sell American options need to be able to obtain the price in a timely manner. This becomes more important when the trader requires the derivative of the option price with respect to the underlying, called the delta of quoted option in order to create a price that is neither too different from competitors and hedge their positions on the option \cite{hull-2003}. This is often done in real time and with limited computing resources. This is what motivated our research into recurrent neural network solutions to the American options problem.

There are numerous classical approaches to solve the American option problem such as the binomial tree method \cite{hull-2003}, numerical solutions to partial differential equations (PDEs) with free boundaries or penalty terms \cite{forsyth-vetzal-2002}, \cite{pironneau-2005}, \cite{duffy-2006}, regression based methods \cite{vanroy-1999}, \cite{longstaff-2001}, \cite{kohler-2010}), stochastic mesh methods \cite{broadie-2004}, etc. The PDE method is the most accurate method to compute option prices and deltas throughout the options horizon. However, when the dimension of an American option, i.e. the number of underlying assets, is greater than 3, numerical solution of PDEs is infeasible, as the complexity grows exponentially with the dimension. When the dimension $d$ is moderate (e.g. $d \leq 20$), the regression-based Longstaff-Schwartz method \cite{longstaff-2001} is widely considered as the state-of-the-art approach for computing option prices. However, Figure 1 in \cite{bouchard-2012} shows that using such regressed values as the spacetime solution is inaccurate. This is reinforced in Figure 8 in \cite{chen-2019}, which shows the clear advantages of using a neural network framework for pricing and hedging.

Outside of classical methods, there is a growing number of applications of neural networks on the American option problem. Some of the early approaches that solves the American option problem with neural networks use simple feedforward networks to approximate the price of the option for restricted number of assets ($d\leq 10$) as seen in \cite{kohler-2010} and \cite{haugh-2004}. Later works extend the scope, and incorporate chained deep neural networks to solve the American option problem \cite{e-2017}, \cite{beck-2018}, \cite{han-2016}, and \cite{fujii-2017}. The work of \cite{guler-2019}, follows the line of work of \cite{han-2016} and \cite{e-2017}, but adds a regularizing term for added stability. The use of a neural network grid has been proposed by \cite{sirignano-2018}. The work of \cite{salvador-2020} uses a neural network to learn option prices of the linear complementary problem of option pricing in an unsupervised manner. More recently, authors of \cite{herrera-2021} proposed a randomized reinforcement learning using recurrent networks to solve the American option problem. However, the methods mentioned do not provide methodology to determine the delta of an American option using neural networks. Delta and price approximation for all spacetime using a deep residual network was proposed by \cite{chen-2019}. In our work we refer to the method of \cite{chen-2019} as the deep residual learning (DRL) method. The work of \cite{hure-2020} provides theoretical convergence results for universal approximators, which RNNs are as shown in \cite{schafer-2007}. The DRL method looks at solving the pricing problem by using neural networks to approximate the solution to stochastic differential equations (SDEs), backward stochastic differential equations (BSDEs) and partial differential equations (PDEs). Their method solves a lot of issues of some earlier work; however, one major drawback is that it is both computationally expensive in time and memory.

In this paper, we propose to use two deep recurrent neural network regression framework to solve high dimensional American style options problems. One of our networks is used to approximate the continuation price and the other is used to approximate the delta. This is in line with the general formulation outlined in \cite{hure-2020}. Our method is used to approximate the solution to the underlying backward stochastic differential equation (BSDE) of high dimensional American style options, it solves for price and delta in all of space and time. The contributions of this paper:
\begin{itemize}
    \item We extend the framework of \cite{chen-2019}, and propose a deep Recurrent neural network (RNN) framework that can be used to solve the American option problem. Using deep RNNs the number of networks that need to be trained is reduced, which means we can speed up the computation time of high dimensional option pricing. We show that our proposed framework achieves a better time and memory complexity for ($d > 20$). Our method computes option prices and deltas faster in the absolute sense compared to the work presented in \cite{chen-2019}, while using much less memory.
    \item We introduce the approximation of deltas in all spacetime using our deep RNN framework, this has not been done by any other RNN option pricing paper to the best of our knowledge. This is done using the pathwise derivative approach which is extended to general time $t$. We incorporate the domain knowledge of American options into our deep RNN framework to approximate accurate, time and memory efficient prices and deltas.
\end{itemize}
In contrast to the approach in \cite{chen-2019}, our work presents a few key differences.
\begin{itemize}
    \item \textit{Different Neural Network Architecture}: The most obvious difference in our approach is the use of \textit{two} deep recurrent network. One deep recurrent network learns the value of the option and the other learns the delta of the the option. By using deep recurrent networks to approximate the price and delta of American option through time, allows us to effectively reduce the number of networks that needs to be trained and stored to two networks as opposed to multiple networks as shown in \cite{chen-2019}.
    \item \textit{Different Optimization criteria}: Our loss function uses an approximation of the continuation function to determine the option value and delta. This is used to train the neural network to determine the value and delta of the option. The delta of the option is determined using the pathwise derivative method at any time $t$.
\end{itemize}

\subsection{Outline of the Paper}
Section \ref{sec:American} outlines the American options problem, Section \ref{sec:BSDE} presents the resulting BSDE formulation and the optimization problem used to solve the American options problem. In Section \ref{sec:RNN} we provide a quick overview on \RNN s, \GRU s and deep recurrent networks. We propose a deep recurrent network regression framework for approximating prices and deltas. We discuss the architecture, feature selection, the approximation algorithm, and complexity of our approach. In Section \ref{sec:computation} we present the runtime and memory complexity of our method compared with the DRL method proposed by \cite{chen-2019}. In Section \ref{sec:Experiments}, we look at pricing and hedging the multidimensional American call options. We compare our method to the method proposed by \cite{chen-2019} and the finite difference method, to show how our proposed method performs compared to the method of \cite{chen-2019}, when we set a runtime limit on the models. We also see cases where our proposed method can reach similar accuracies to the DRL method in shorter runtime. Section \ref{sec:Conclusion} concludes our paper and outlines future work in neural network applications in option pricing.

\section{American Options}\label{sec:American}

In this paper, we use capital and lowercase letters to distinguish random and deterministic variables respectively. Suppose we have a basket of $d$ stocks; with price processes $\vv{S}(t) = [S_1(t), ..., S_d(t)]^\top \in \mathbb{R}^d$. Note that $S$ is a random variable and $\vv{S}$ is a vector of random variables. Furthermore, in order to distinguish vectors from scalars we will use an overline for vectors and bold text for matrices. Let $t \in [0,T]$ be the time up to maturity $T$, and let $r$ be the interest rate. Let $\delta_i$, $\mu_i = (r-\delta_i)$ and $\sigma_i$ be the dividend, drift and volatility of stocks $i=1,...,d$. Let $\mathbf{\rho} \in \mathbb{R}^{d\times d}$ be the correlation matrix and we define the correlated random variable $dW^i(t) = \sum\limits_{j=1}^{d}L_{i,j}\phi_j(t) \sqrt{dt}$, where $\phi_j(t)\sim N(0,1)$ are i.i.d, and $\mathbf{L}$ is the Cholesky factor of $\bm\rho$, i.e. $\bm\rho=\mathbf{LL}^\top$. Given a vector of initial prices, $\vv{s}(0)\in\mathbb{R}^d$, the price vector $\vv{S}(t)$ follows the dynamics given by geometric Brownian motion. For each element $i$ we have:
\begin{equation}
    dS_i(t) = \mu_i S_i(t) dt + \sigma_i S_i(t) dW_i(t),
    \label{eq:SDE}
\end{equation}
with initial condition $S_i(0) = s_i(0)$.
For simplicity, we assume that the market has sufficient liquidity and we do not consider market frictions.
Let $f(\vv{s}(t))$ be the payoff of the American option at the realized state $\vv{s}(t)$. This payoff function is a function of $g(\vv{s}(t))$, where $g(\vv{s}(t))$ varies depending on the payoff style, i.e. the commonly seen max call option has a payoff of the form $\max\limits_{i=1,...,d}\{s_i(t)-K\}$, and the general payoff has the form:
\begin{equation}
    f(\vv{s}(t)) := \max\{g(\vv{s}(t)), 0\}.
    \label{eq:payoff}
\end{equation}

The American option is more complicated to price than the European option with the same payoff, because it can be exercised at any point $t$. To determine if we want to exercise or not, we need the continuation price.

Let $c(\vv{s}(t))$, be the continuation price of the American call option, i.e. the discounted payoff of the option when it is not exercised at time $t$ at price $\vv{s}(t)$. We use $\tau=[t,T]$ to denote the optimal stopping time; which is the best time to exercise the contract between time $t$ and maturity $T$. Then we can define the continuation price as:
\begin{equation}
    c(\vv{s}(t)) = \max\limits_{\tau \in [t,T]}\mathbb{E}[e^{-r (\tau - t)}f(\vv{S}(\tau))|\vv{S}(t) = \vv{s}(t)].
\label{eq:continuation}
\end{equation}
In otherwords, we exercise the option when we determine that the continuation price is the same as the payoff. Then the value of the American style option is given by:
\begin{align}
    v(\vv{s}(t),t) &= \max\{f(\vv{s}(t)),\ c(\vv{s}(t))\} \nonumber \\
    &= \begin{cases} 
      f(\vv{s}(t),t) & f(\vv{s}(t),t)\geq c(\vv{s}(t)) \\
      c(\vv{s}(t)) & \textnormal{otherwise} 
   \end{cases},
   \label{eq:AmOp}
\end{align}
and the delta of the American style option is given by:
\begin{equation}
    \nabla v(\vv{s}(t),t) = \begin{cases} 
      \nabla f(\vv{s}(t),t) & f(\vv{s}(t),t)\geq c(\vv{s}(t)) \\
      \nabla c(\vv{s}(t)) & \textnormal{otherwise} 
   \end{cases}.
   \label{eq:AmOp_delta}
\end{equation}
The optimal exercise boundary can be found by using the continuation price such that
\begin{equation}
    c(\vv{s}(t)) = f(\vv{s}(t),t),
    \label{eq:equal}
\end{equation}
where $c(\vv{s}(t))$ is found by solving the optimization problem \eqref{eq:continuation}.

In practical applications of hedging, we are interested in the delta of the American option. The delta of an option is the first derivative of the value of the option with respect to the stock price, $\nabla v(\vv{s}(t), t) = [\frac{\partial v}{\partial s_1}(\vv{s}(t),t), ... \frac{\partial v}{\partial s_d}(\vv{s}(t),t)]^\top$. One of the advantages of our approach is that it can solve both the option price $v(\vv{s}(t),t)$ and delta $\nabla v(\vv{s}(t), t)$ on the entire space-time in a more time efficient way compared to \cite{chen-2019}.

\section{Backward Stochastic Differential Equation(BSDE) Formulation} \label{sec:BSDE}

\subsection{BSDE Formulation}

We convert the American option problem into a BSDE using the following theorem:
\begin{theorem}(BSDE formulation)
Assume that an American option is not exercised at time $[t,t+dt]$. Then the price of an American option at time $t$ satisfies the following BSDE
\begin{equation}
    dc(\vv{S}(t)) = r c(\vv{S}(t))dt+\sum\limits_{i=1}^{d}\sigma_i\frac{\partial c}{\partial s_i}(\vv{S}(t))dW_i(t).
    \label{eq:BSDE}
\end{equation}
where $\vv{S}(t)$ satisfies \eqref{eq:SDE} and $r$, $\sigma_i$ and $dW_i(t)$ are defined as in \eqref{eq:SDE}.
\end{theorem}
\begin{proof}
Interested readers to the proof are referred to \cite{chen-2019}.
\end{proof}
In this paper, we solve the BSDE for the option price $v(S(t),t)$ and the delta $\frac{\partial v}{\partial S_i}(S_i(t),t)\}$ at the exercise boundary. Using \eqref{eq:BSDE} and \eqref{eq:equal} we get the following integral form of the BSDE that models the option value at the exercise boundary:
\begin{equation}
    c(\vv{S}(t)) = f(\vv{S}(T)) + \int\limits_{t}^{T}e^{-rq}c(\vv{S}(q))dq - \int\limits_{t}^{T} e^{-rq} \sum\limits_{i=1}^{d}\frac{\partial c}{\partial S_i}(S_i(q)) \sigma_i dW_i(q).
\label{eq:value}
\end{equation}

The significance of the BSDE formulation is two fold; one, it couples the option value, $v(\vv{S}(t))$, and the delta, $\nabla v(\vv{S}(t))$ \cite{chen-2019}. If the approximation for the value is correct then the approximation for the delta will also be correct, which allows us to perform a complete hedging process. The second significance of the BSDE formulation is that it allows us to design a less expensive neural network compared to the Hamilton-Jacobi-Bellman (HJB) partial differential equation (PDE) as shown in \cite{sirignano-2018}. This is because using the HJB PDE requires us to compute and store the Hessian tensor, which is an $\mathcal{O}(Md^2)$ tensor, where $M$ is the number of samples used to train the neural network. Compared to the HJB formulation, the BSDE formulation only requires the Jacobian tensor of size $\mathcal{O}(Md)$.

\subsection{Continuation price and delta at time $t$}

In our application we need to estimate the continuation price and delta at $t$. Using the BSDE formulation we can solve for the continuation price and delta of the option across all space and time. The authors of \cite{chen-2019} shows that approximating the continuation value using the discounted expected value from the previous timestep is reasonable, the expected value of the discounted continuation price is given as
\begin{equation}
    c(\vv{s}_t,t) = \mathbb{E}[e^{-r\Delta t}v(\vv{S}_{t+\Delta t})].
    \label{eq:value0}
\end{equation}
However, we found that this is only the case at maturity, for general approximation of the continuation price at time $t$, it is better to solve the optimization problem \eqref{eq:continuation}, which gives us a continuation price at $t$ as:
\begin{equation}
    c(\vv{s}(t)) = \mathbb{E}[e^{-r (\tau - t)}f(\vv{S}(\tau))].
\label{eq:continuation_discrete}
\end{equation}
The work of \cite{broadie-1996} shows a framework for estimating the delta of a European call options using the pathwise derivative and was extended to the American option in \cite{thom-2009}. The pathwise derivative adapted for American options is shown in \cite{chen-2019} for $t=0$ and is given as
\begin{equation}
    \frac{\partial v}{\partial s_i}(\vv{s}_0,0) = \mathbb{E}[e^{-r\tau}\frac{\partial f}{\partial S_i}(\vv{S}(\tau))\mathds{1}_{f\geq c}\frac{\partial S_i}{\partial s_i(0)}].
    \label{eq:pthdelta0}
\end{equation}
Where the function $\mathds{1}_{f\geq c}$ is equal to $1$ when the payoff $f$ is greater than the continuation value $c$ and $0$ otherwise at $\tau$. We can see in \cite{broadie-1996} that when the underlying asset evolves following the process \eqref{eq:SDE}, $\frac{\partial S_i}{\partial s_0} = \frac{S_i}{s_0}$. We can extend this notion to the increment $[t,\tau]$ under the same process since
\begin{equation}
    \frac{\partial S_i(\tau)}{\partial s_i(t)} = e^{(\mu-\frac{\sigma^2}{2})(\tau - t) + \sigma\sqrt{(\tau - t)}W_{t}} = \frac{S_i(\tau)}{s_i(t)}.
    \label{eq:del_price}
\end{equation}
In our method we use this concept to approximate the delta of the continuation price at anytime $t$. Let $\frac{\partial c}{\partial s_i}$ be the delta of the continuation price then we estimate the delta using \eqref{eq:del_price} and is given by
\begin{equation}
    \frac{\partial c}{\partial s_i}(\vv{s}_i(t),t) = \mathbb{E}[e^{-r(\tau - t)}\frac{\partial f}{\partial S_i}(\vv{S}(\tau))\frac{\partial S_i(\tau)}{\partial s_i(t)}].
    \label{eq:pthdelta}
\end{equation}
The pathwise derivative approach may not be applicable if the derivative $\partial S_i /\partial s_i$ cannot be evaluated. This may occur if the underlying asset does not follow the process \eqref{eq:SDE}. However, in our framework, we are always able to evaluate \eqref{eq:pthdelta}. This is because our framework uses a deep RNN to approximate the continuation price of the option, and using autodifferentiation of the network with respect to the underlying asset price we can always evaluate the derivative $\partial c / \partial s_i(t)$.

\subsection{Discretization of BSDE}

In this paper, we use the Euler-Maruyama method as in \cite{chen-2019} and \cite{hure-2020} to simulate SDE \eqref{eq:SDE} and BSDE \eqref{eq:value}. Let $m=1,..,M$ be the indices of simulation paths and $n=1,...,N$ be the indices of the discrete timesteps from $0$ to $T$. Let $\Delta t = T/N$ be the timestep size and let $t^n = n\Delta t$. We discretize $dW_i(t)$ as 
\begin{equation*}
    (\Delta W_i)_m^n = \sum\limits_{j=1}^{d}L_{i,j}(\phi_j)_m^n\sqrt{\Delta t}.
\end{equation*}
We discretize \eqref{eq:SDE} as
\begin{align}
    (S_i)_m^0 &= (s_i)_m^0,\ &&i=1,..,d\\
    (S_i)_m^{n+1} &= (1+(r-\delta_i)\Delta t)(S_i)_m^n + \sigma_i(S_i)_m^n(\Delta W_i)_m^n,\ &&i=1,...,d.
    \label{eq:sdedis}
\end{align}
We can discretize the payoff function at $n+1$ as
\[
    f(\vv{S}_m^{n+1}) = \max\{g(\vv{S}_m^{n+1}),0\}
\]
For $n=N-1,...,0$, we discretize \eqref{eq:value} as 
\begin{equation}
    c(\vv{S}_m^{n+1},t^{n+1}) = (1+r\Delta t)c(\vv{S}_m^n,t^{n}) + \sum\limits_{i=1}^{d}\sigma_i(S_i)_m^n\frac{\partial c}{\partial s_i}(\vv{S}_m^n,t^n)(\Delta W_i)_m^n,
    \label{eq:bsdedis}
\end{equation}
and the value of the option is discretized as
\[
    v(\vv{S}_m^{n+1},t^{n+1}) = \max\{f(\vv{S}_m^{n+1}),c(\vv{S}_m^{n+1},t^{n+1})\}
\]
with terminal condition:
\begin{equation}
    v(\vv{S}_m^N,t^N) = f(\vv{S}_m^N).
    \label{eq:term}
\end{equation}
Then we solve \eqref{eq:bsdedis} for $c(\vv{S}_m^n,t^{n}),\partial c(\vv{S}_m^n,t^{n})/\partial s_i$ using the least squares solution presented in section \ref{sec:LS-BSDE}. This is then used to compute the option price, \eqref{eq:AmOp} and delta, \eqref{eq:AmOp_delta}. This allows to account for the possibility of early exercise as we compare the payoff and the continuation value.

We solve the discretized BSDE by using \eqref{eq:sdedis} to generate samples of the underlying asset prices $\{\vv{S}_m^n\}$ for all time indexed by $n$ and samples indexed by $m$. Then, starting at $n=N$, we compute the value \eqref{eq:term} and compute $v(\vv{S}_m^n,t^n)$ by solving \eqref{eq:value} backwards for each time $n=N-1,...,0$. We then obtain the option price and delta at $t=0$ using our neural network output at $t=0$.

\subsection{Least squares solution to the BSDE}\label{sec:LS-BSDE}

We consider the $n$-th timestep $t^n$, and introduce a short notation for the value function, $v^n(\vv{s})\equiv v(\vv{s},t^n)$, and the delta function, $\nabla v^n(\vv{s})\equiv \nabla v(\vv{s},t^n)$. We use a short notation for the payoff function $f^n(\vv{s})\equiv f(\vv{s},t^n)$ and define $\nabla f^n(\vv{s}) = [\frac{\partial f^n}{\partial s_1}(\vv{s}), ..., \frac{\partial f^n}{\partial s_d}(\vv{s})]^\top$. Solving \eqref{eq:bsdedis} requires us to find $d$ dimensional functions $v^{*n}(\vv{s})$ and $\nabla v^{*n}(\vv{s})$, such that they satisfy \eqref{eq:bsdedis}. When $d=1$, this can be solved using closed form solutions, but for $d \geq 2$ a closed form solution does not exist.

The loss function, presented in \cite{chen-2019} minimizes the residual $R^n$ between $v^{n+1}$ and \eqref{eq:bsdedis} given by 
\begin{equation*}
    R^n = v^{n+1} - (1+r\Delta t)c(\vv{S}(t))+\sum\limits_{i=1}^{d}\sigma_i\frac{\partial c}{\partial s_i}(\vv{S}(t))\Delta W_i(t).
\end{equation*}

In this paper we present a different loss function which allows us to learn the price and delta function. Let $y^n$ and $\nabla y^n$ be the approximation for the option price and delta, then we define our loss function, $\mathcal{L}$, as 
\begin{equation}
    \mathcal{L}[y^n,\nabla y^n] =  \mathbb{E}[c^{n}-y^n|^2] + \Delta t\mathbb{E}[|\nabla c^{n}\cdot\tilde\sigma - \nabla y^n\cdot\tilde\sigma|^2].
    \label{eq:loss}
\end{equation}
where $\tilde\sigma = \sigma\times\bm\rho$, the symbol $\times$ is the cross product and the symbol $\cdot$ denotes the dot product. The goal of \eqref{eq:loss} is to find the approximation $y^n$ and $\nabla y^n$ to the option price and the option delta. The authors of \cite{hure-2020}, show that the loss function \eqref{eq:loss} used in our approach convergence to the viscosity.

Note, if we denote $\tau$ as $\tau = \tilde{n}\Delta t$ then the continuation price can be approximated using \eqref{eq:continuation_discrete} as:
\begin{equation}
    c^n \approx \mathbb{E}[e^{-r(\tilde{n} - n)\Delta t}f^{\tilde{n}}].
    \label{eq:cont_approx}
\end{equation}
Assuming the process follows the process in \eqref{eq:SDE}\footnote{For the case where we do not follow the process \eqref{eq:SDE}, we can use autodifferentiation to determine the delta, $\partial c^{n+1}/\partial S_i^{n}$.} we can, apply the pathwise derivative method, to approximate the continuation delta using \eqref{eq:pthdelta} as:
\begin{equation}
    \nabla c^n \approx \mathbb{E}[e^{-r(\tilde{n} - n)\Delta t}\nabla f^{\tilde{n}}\frac{S^{\tilde{n}}}{S^{n}}].
    \label{eq:cont_delta_approx}
\end{equation}
Using \eqref{eq:cont_approx} and \eqref{eq:cont_delta_approx} we can determine the $v^n$ and $\nabla v^n$ using \eqref{eq:AmOp}-\eqref{eq:AmOp_delta}. We use the loss \eqref{eq:loss} over loss of \cite{chen-2019} because:
\begin{itemize}
    \item the loss \eqref{eq:loss} does not require the integration over time steps, reducing the need to store the results of each time step. This also allows us to train the recurrent network for all time in one step.
    \item In our approach the delta is approximated by a neural network, which allows us to calculate delta at $t=0$ for any stopping time $\tau$ using the pathwise derivative method without having to worry if $\partial v / \partial s_0$ is evaluable.
\end{itemize}

Since the BSDE is a backward process, we start our computation at $n=N$, where the value of the option is given by the payoff \eqref{eq:payoff}. This means that the risk-neutral price of the option at $n=N-1$, needs to be the discounted value of the option. To train our network, we must find the price and delta at time $n$ such that the value of the option at $n$ is given by \eqref{eq:AmOp}, where the option price is approximated by our deep RNN. Then for timestep $n$ our value is the output of our deep RNN for computing prices, $y^n$. Similarly, the delta at $n$ is the output of our deep RNN for computing deltas, $\nabla y^n$.

We evaluate the loss as the expected value over $m$ simulations such that the approximations $y^n$ and $\nabla y^n$ minimizes \eqref{eq:loss}
\begin{equation}
    \{{y}^{*n},\nabla {y}^{*n}\} = \textnormal{arg}\min\limits_{y^n,\nabla y^n}[\mathcal{L}[{y}^n,\nabla {y}^n]].
    \label{eq:optprob}
\end{equation}

\section{Deep Recurrent Neural Network Formulation}\label{sec:RNN}

To determine the approximation functions in \eqref{eq:optprob} is non-trivial. To solve this problem, we can use some parameterized function, such as polynomials, to approximate $y^n$, which allows us to solve the optimization problem in the parameter space. Classical methods such as the Longstaff-Schwartz method is built on this principle; however, it has some drawbacks as mentioned in \cite{chen-2019}. More modern methods as seen in \cite{herrera-2021} can be used to alleviate some this by using neural networks to learn such functions. We note that, methods such as the Longstaff-Schwartz method and the method proposed by \cite{herrera-2021} do not solve the BSDE problem outlined in \eqref{eq:BSDE}.

A main drawback of the deep neural network approach in \cite{chen-2019} is that they require a feedforward neural network for every time step. It may be somewhat improved by skipping every $J$ time timsteps, but it still requires to train $N/J$ networks where $N$ is the total number of timesteps. When $N$ is large, the training time can be very costly. 

Our proposed approach is to use a deep recurrent network to compute the approximate option price $y^n$ and delta $\nabla y^n$. The advantage of our deep recurrent network approach is that we only require two networks. One is to used to compute $y^n$ and the other for $\nabla y^n$ over all time and space. This leads to a tremendous saving in time and memory for large $N$.

\subsection{Sequence of Neural Networks}

Our approach is to use two deep RNNs to represent the approximate option price function, $y^n$ and the option delta function, $\nabla y^n$. The main advantage of using a Recurrent network over traditional feedforward networks is that the network is able to learn time dependent dynamic problems \cite{elman-1990}. There exists many different types of RNNs and we refer readers to \cite{cho-2014} and \cite{hochreiter-1997} to review standard GRUs and LSTMs. In our approach we augment traditional GRUs to solve the American option problems.

\begin{figure}[!htbp]
    \centering
    \includegraphics[width=\textwidth]{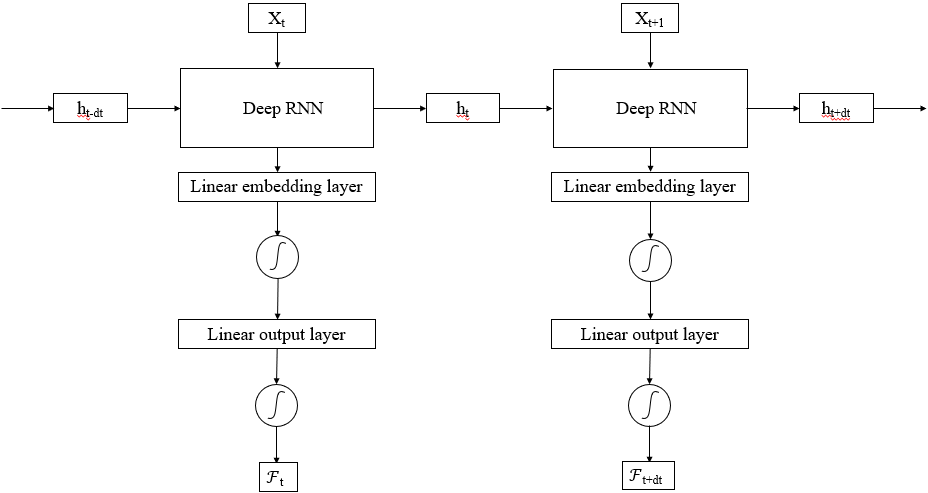}
    \caption{Deep Recurrent network for price rolled out through time. Our network is composed of a recurrent unit, an embedding linear layer and an output linear layer. Takes in a matrix of inputs $\mathbf{x}$, hidden information $\mathbf{h}$ and outputs $\mathcal{F}$.}
    \label{fig:network}
\end{figure}

\begin{figure}[!htbp]
    \centering
    \includegraphics[width=0.65\textwidth]{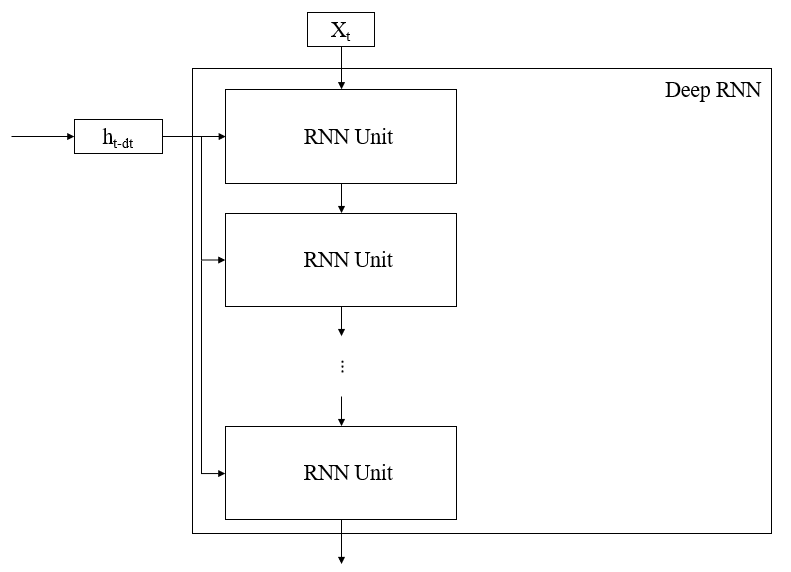}
    \caption{Stacked RNN units within the deep RNN}
    \label{fig:network_units}
\end{figure}

In this paper we present two recurrent networks, rolled out through time, $\{y^n(\mathbf{s};\Omega)| n = N-1,...,0\}$ and $\{\nabla y^n(\mathbf{s};\Theta)| n = N-1,...,0\}$. The trainable set of parameters $\{\Omega,\Theta\}$ of the network approximates option price and delta for each time step $n$. We will first present our recurrent network used to approximate price, $\{y^n(\mathbf{s};\Omega)| n = N-1,...,0\}$.

The arbitrage free price of an option should differ from the previous price and delta by a function with magnitude $\mathcal{O}(\Delta t)$. As shown in \cite{chen-2019}, using a linear combination of the approximations $y^n$ and the discounted $y^{n+1}$ provides an estimate to the continuation price with lower variation, which improves the accuracy of the regression. 

We note that our problem is solved backwards in time, and we use ``previous", ``current" and ``next" to refer to the $(n+1)^\textnormal{th}$, $n^\textnormal{th}$, and $(n-1)^\textnormal{th}$ timesteps respectively. In our method, the approximation at time $n = N-1,...,0$, is a linear combination of the discounted approximation from the $(n+1)^{th}$ timestep and the $n^{th}$ approximation from the deep RNN, \cite{chen-2019}, has shown that estimating the continuation price in this manner is more stable. We introduce the trainable coefficient $\alpha \in [0,1]$, for $y^n$. The coefficient $\alpha$ determines how much of the approximation is given by the discounted $y^{n+1}$ and the deep RNN approximation $y^n$. Mathematically, we approximate the continuation price as
\begin{align}
    y^N(\mathbf{s}) &= f^N(\mathbf{s})\nonumber \\
    y^n(\mathbf{s};\Omega) &= (1-\alpha)e^{-r \Delta t}c^{n+1}+\alpha \mathcal{F}(\mathbf{s};\Omega^{n}),\ &&n=N-1,...,0
\end{align}
where $\mathcal{F}(\mathbf{s};\Omega)$ is the network approximation of the price. The set of parameters $\Omega$ includes all the weights and biases used in Algorithm \ref{alg:price_network}.

Our deep RNN to compute continuation price, $\mathcal{F}(\mathbf{s};\Omega)$, is shown in Figure \ref{fig:network}. We parameterize it by an $L$-layer deep RNN with GRU units with a linear embedding layer and output layer as shown in Figure \ref{fig:network_units}. A typical GRU unit is illustrated in Figure \ref{fig:gru_scheme}. Let the input of the deep RNN be $\mathbf{X} \in \mathbb{R}^{M \times 2d}$ and previous hidden state be $\mathbf{h}^{n-1} \in \mathbb{R}^{M \times 2d}$. The GRU unit is composed of the update gate $\mathbf{z}^n$ which decides if the a new hidden state should be stored, the reset gate $\mathbf{r}^n$, decides which information to keep from the previous sequence, the new hidden state $\hat{\mathbf{h}}^n$ stores potential information from a filtered previous state and current state. Mathematically the GRU is given by the set of equations
\begin{align}
        \mathbf{r}^n &= \hat\sigma(\mathbf{W}_r^{n\top} [\mathbf{h}^{n+1}, \mathbf{X}^n] + b_r)\label{eq:gru_eq_start}\\
        \mathbf{z}^n &= \hat\sigma(\mathbf{W}_z^{n\top} [\mathbf{h}^{n+1},\mathbf{X}^n] + b_z)\\
        \hat{\mathbf{h}}^n &= \textnormal{tanh}(\mathbf{W}_{\hat{h}}^{n\top} [\mathbf{r}^n \odot \mathbf{h}^{n+1},\mathbf{X}^n] + b_{\hat{h}})\\
        \mathbf{h}^n &= (1-\mathbf{z}^n)\odot \mathbf{h}^{n+1} + \mathbf{z}^n \odot \hat{\mathbf{h}}^n,
    \label{eq:gru_eq_end}
\end{align}
where $\hat\sigma$, tanh are the typical sigmoid and hyperbolic tangent activation functions seen in \cite{goodfellow-2016}. The parameters $\mathbf{W}^n_r,\mathbf{W}^n_z, \mathbf{W}^n_h \in \mathbb{R}^{d\times d}$ are trainable weights and $b_r,b_z,b_h \in \mathbb{R}^{d}$ are trainable offsets.

In the following part, we temporarily use a subscript square bracket for the layer index $l=1,...,L$. We construct the $L$-layer deep RNN for continuation price as
\begin{itemize}
    \item For layers, $l=1$:
    \begin{align*}
        \mathbf{r}^n_{[l]} &= \hat\sigma(\mathbf{W}_r^{n\top} [\mathbf{h}^{n+1}, \mathbf{X}^n] + b_r)\\
        \mathbf{z}^n_{[l]} &= \hat\sigma(\mathbf{W}_z^{n\top} [\mathbf{h}^{n+1},\mathbf{X}^n] + b_z)\\
        \hat{\mathbf{h}}^n &= \textnormal{tanh}((\mathbf{W}_{\hat{h}}^{n\top})_{[l]} [\mathbf{r}^n_{[l]} \odot \mathbf{h}^{n+1},\mathbf{X}^n] + b_{\hat{h}})\\
        \mathbf{h}^n_{[l]} &= (1-\mathbf{z}^n_{[l]})\odot \mathbf{h}^{n+1} + \mathbf{z}^n_{[l]} \odot \hat{\mathbf{h}}^n_{[l]}.
    \end{align*}
    \item For layers, $l=2,...,L$:
    \begin{align*}
        \mathbf{r}^n_{[l]} &= \hat\sigma(\mathbf{W}_r^{n\top} [\mathbf{h}^{n+1}, \mathbf{h}^n_{[l-1]}] + b_r)\\
        \mathbf{z}^n_{[l]} &= \hat\sigma(\mathbf{W}_z^{n\top} [\mathbf{h}^{n+1},\mathbf{h}^n_{[l-1]}] + b_z)\\
        \hat{\mathbf{h}}^n &= \textnormal{tanh}((\mathbf{W}_{\hat{h}}^{n\top})_{[l]} [\mathbf{r}^n_{[l]} \odot \mathbf{h}^{n+1},\mathbf{h}^n_{[l-1]}] + b_{\hat{h}})\\
        \mathbf{h}^n_{[l]} &= (1-\mathbf{z}^n_{[l]})\odot \mathbf{h}^{n+1} + \mathbf{z}^n_{[l]} \odot \hat{\mathbf{h}}^n_{[l]}.
    \end{align*}
    \item For the embedding layer:
    \begin{equation*}
        \mathbf{e}^n = \textnormal{swish}(\mathbf{W}^{n\top}_{e} \mathbf{h}^n_{[L]} + b^{n}_{e}),
    \end{equation*}
    where $\mathbf{W}^{n}_{e} \in \mathbb{R}^{d\times d}$ are trainable weights and $b^{n}_{e} \in \mathbb{R}^{d}$ is the trainable bias.
    \item For the output layer:
    \begin{equation*}
        \mathcal{F}(\mathbf{s};\Omega^{n}) = \textnormal{softplus}(\mathbf{W}^{n\top}_{\mathcal{F}} \mathbf{e}^n + b^{n}_{\mathcal{F}}),
    \end{equation*}
    where $\mathbf{W}^{n}_{\mathcal{F}} \in \mathbb{R}^{d\times d}$ are trainable weights and $b^{n}_{\mathcal{F}} \in \mathbb{R}^{d}$ is the trainable bias. The swish activation function is given in \eqref{eq:swish}
\end{itemize}

We construct the deep RNN to compute the deltas, $\mathcal{G}(\mathbf{s};\Theta)$ in a similar fashion to $\mathcal{F}(\mathbf{s};\Omega)$. The main difference between the two networks is in the output layers where
\begin{equation*}
    \mathcal{G}(\mathbf{s};\Theta^{n}) = \hat\sigma(\mathbf{W}^{n\top}_{\mathcal{G}} \mathbf{e}^n + b^{n}_{\mathcal{G}}),
\end{equation*}
where $\mathbf{W}^{n}_{\mathcal{G}} \in \mathbb{R}^{d\times d}$ are trainable weights and $b^{n}_{\mathcal{G}} \in \mathbb{R}^{d}$ is the trainable bias. The sigmoid activation function is given in \eqref{eq:sigmoid}. The sigmoid function is used here as it is the derivative of the softplus function with respect to the input.

To compute deltas, we introduce the trainable coefficient $\beta \in [0,1]$. $\beta$, like $\alpha$ determines how much of the approximation is given by the discounted $\nabla y^{n+1}$ and the deep RNN approximation $\nabla y^n$. Mathematically, we approximate the option delta as
\begin{align}
    \nabla y^N(\mathbf{s}) &= \nabla f^N(\mathbf{s})\nonumber \\
    \nabla y^n(\mathbf{s};\Theta) &= (1-\beta)e^{-r \Delta t}\nabla c^{n+1}\frac{S^{n+1}}{S^{n}}+\beta \mathcal{G}(\mathbf{s};\Theta^n),\ &&n=N-1,...,0.
\end{align} 
The set of parameters $\Theta$ includes all the weights and biases used in Algorithm \ref{alg:delta_network}. 

Our deep Recurrent networks are structured as a one way stacked recurrent units as shown in Figure \ref{fig:network_units}. In our method, we use $L=7$ units. For $n = N-1,...,0$ our input, $\mathbf{X}^{n}$ passes through a the first Recurrent unit. The subsequent Recurrent units use the output of the previous unit as the input of the next, while the previous hidden state, $\mathbf{h}^{n+1}$, is passed to all the recurrent units. The output of the deep RNN is passed through linear embedding layer an activation function then a linear output layer, and a final activation function which is chosen based on the domain information of the problem.

\subsection{Recurrent unit selection}\label{rns}

A recurrent network is a sequential learning model based on neural networks \cite{jordan-1990}. The key differentiating point of a recurrent network as opposed to a neural network is the use of hidden states \cite{elman-1990}. The hidden states of a recurrent network, allows us to retain information as the network is trained over a sequence.

The GRU is an extension of RNN and alleviates the problem of vanishing/exploding gradients \cite{cho-2014}. This is done by feeding through a combination of the previous state and the current state. The GRU is a variation of a popular RNN called the LSTM (Long short term memory), and uses gates (activation functions) to transform the input. The GRU reduces the complexity of the LSTM by reducing the number of gates that the information passes through. 

In this work, we use the Gated Recurrent Unit (GRU) as our RNN unit as shown in Figure \ref{fig:network}. The $\sigma$ and tanh are the sigmoid and hyperbolic tangent activation functions as seen in \cite{goodfellow-2016}. The linear structure of the GRU output allows us to generalize the payoff and delta of the option with better accuracy than non-linear structure of the LSTM. Past studies have shown that linear structure of the output preserves the most information as observed in \cite{chen-2019} and studied in \cite{he-2016}. We found that the non-linear output of LSTMs can force the output to achieve extreme negative values, which is not ideal for pricing. Another reason is that it has less parameters to optimize as it only has three gates over the LSTMs four gates. This reduces the number of operations we need to perform and the number of parameters that need to be optimized, which further increases our computational efficiency. We refer interested readers to \cite{cho-2014} for a more detailed review of GRUs and \cite{hochreiter-1997} for a review of LSTMs. 

\begin{figure}[!htbp]
    \centering
    \includegraphics[width=\textwidth]{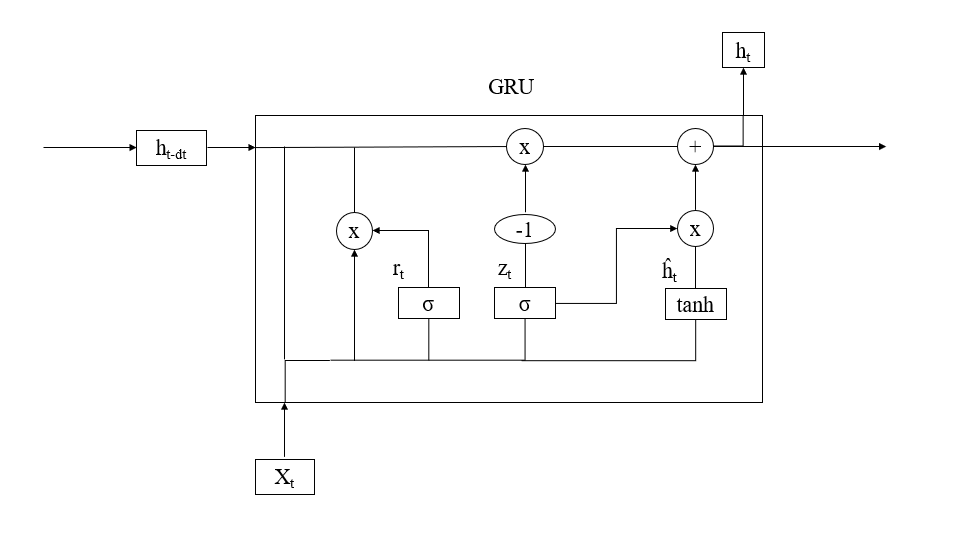}
    \caption{Typical GRU cell with gates. Operations in circles are element wise operations.}
    \label{fig:gru_scheme}
\end{figure}

\subsection{Activation function selection for output and embedding layers}

The choice of suitable activation functions leads to better predictive capabilities for any network \cite{goodfellow-2016}. As shown in Figure \ref{fig:network}, our network adds two linear layers to the output of the deep recurrent network, a linear embedding layer and a linear output layer. We denote the embedding layer as $\mathbf{e}^n$, $\mathbf{o}_p^n$ for the output layer of the price approximation and $\mathbf{o}_\Delta^n$ for the output layer of the delta approximation. For continuation prices the natural activation function used for the linear embedding layer and output layer is the popular rectilinear unit (ReLU) activation function \cite{goodfellow-2016}. The ReLU function has the form
\begin{equation*}
    \textnormal{ReLU}(x) = \max\{x,0\},
\end{equation*}
which has the same form as our payoff function $f(x)$. However, the ReLU activation function suffers from vanishing gradients if a network is trained for too many iterations. Thus instead, we use the softplus activation function given by
\begin{equation*}
    \textnormal{softplus}(x,\kappa)=\frac{1}{1+e^{-\kappa x}}.
\end{equation*}
This activation function approaches the ReLU activation as $\kappa \rightarrow \infty$. This gives us a good generalization of the payoff and it's exponential behaviour as it evolves through time. For the output layers of the network that approximates the option deltas we use a sigmoid activation function. This selection is natural as the delta is the derivative of the value of the option with respect to price.

For the embedding layer, we use the swish activation function given by
\begin{equation}
    \textnormal{swish}(x) = x\hat\sigma(x),
    \label{eq:swish}
\end{equation}
where $\hat\sigma(x)$ is the sigmoid activation function given by
\begin{equation}
    \hat\sigma(x) = \frac{1}{1+e^{-x}},
    \label{eq:sigmoid}.
\end{equation}

The swish activation function is empirically shown to perform better than the ReLU function, as shown in \cite{rama-2017} and helps our method reach more accurate solutions.

\subsection{Feature selection and network summary} \label{fsns}

Feature selection in neural network has a significant impact on the accuracy of the network model \cite{goodfellow-2016}, the correct choice of inputs in any regression problem allows for good predictions and approximations. We use the asset price, and $g(\mathbf{s})$ from \eqref{eq:payoff}. We found that $g(\mathbf{s})$ performs better than using the payoff, $f(\mathbf{s})$ due to the additional information it carries. We concatenate the inputs but we do not perform normalization. Our concatenated input matrix is
\begin{equation*}
    \mathbf{x} = [\mathbf{s},g(\mathbf{s})]^{T}\ \in\ \mathbb{R}^{M \times 2d}.
\end{equation*}
Unlike a traditional feedforward network the Recurrent network uses the hidden state to carry information through time. In practical applications of Recurrent networks the hidden state is initialized as zero or given some random numbers following a uniform or normal distribution. In our approach, we initialize the hidden state using $\mathbf{y}^N$ for $\mathcal{F}(\mathbf{s};\Omega)$ and $\nabla \mathbf{y}^N$ for $\mathcal{G}(\mathbf{s};\Theta)$. We pad the hidden state with itself to match the size of the input and for each timestep $n$, we redefine the hidden states as $\mathbf{y}^n$ for $\mathcal{F}(\mathbf{s};\Omega)$ and $\nabla \mathbf{y}^n$ for $\mathcal{G}(\mathbf{s};\Theta)$.

To summarize sections (\ref{rns}) - (\ref{fsns}), we can express our price approximating network for $n=N-1,...,0$ by Algorithm \ref{alg:price_network} and similarly we approximate the delta of the option we use Algorithm \ref{alg:delta_network}.

\begin{algorithm}[!htbp]
    \caption{Continuation price network using a single GRU unit}
    \label{alg:price_network}
    \begin{algorithmic}
        \Require{$\mathbf{h}^0,\mathbf{W}_e^n,\mathbf{W}_o^n,\mathbf{c}^{n+1}$}
        \For{$n=N-1,...,0$}
            \Let{GRU}{\eqref{eq:gru_eq_start}-\eqref{eq:gru_eq_end}}
            
            \Let{$\mathbf{e}^n$}{$ \textnormal{swish}(\mathbf{W}_e^{n\top}\textnormal{GRU} + b_e)$}
            
            \Let{$\mathbf{o}_p^n$}{$ \textnormal{softplus}(\mathbf{W}_o^{\top}\mathbf{e}^{n} + b_o)$}
            \Let{$\mathcal{F}(s;\Omega)$}{$\mathbf{o}_p^n$}            \Let{$\mathbf{y}^n$}{$(1-\alpha)e^{-r \Delta t}\mathbf{c}^{n+1} + \alpha \mathcal{F}(\mathbf{s};\Omega)$}
        \EndFor
    \end{algorithmic}
\end{algorithm}
\begin{algorithm}[!htbp]
    \caption{Delta network using a single GRU unit}
    \label{alg:delta_network}
    \begin{algorithmic}
        \Require{$\mathbf{h}^0_\Delta, \mathbf{W}_e^n,\mathbf{W}_o^n,\nabla \mathbf{c}^{n+1}$}
        \For{$n=N-1,...,0$}
            \Let{ GRU}{\eqref{eq:gru_eq_start}-\eqref{eq:gru_eq_end}}
            
            \Let{$\mathbf{e}^n$}{$ \text{swish}(\mathbf{W}_e^{n\top}\textnormal{GRU} + b_e)$}
            
            \Let{$\mathbf{o}_\Delta^n$}{$ \hat\sigma(\mathbf{W}_o^{\top}\mathbf{e}^{n} + b_o)$}
            \Let{$\mathcal{G}(s;\Theta)$}{$\mathbf{o}_\Delta^n$}            \Let{$\nabla\mathbf{y}^n$}{$(1-\beta)e^{-r \Delta t}\nabla\mathbf{c}^{n+1}\frac{\mathbf{s}^{n+1}}{\mathbf{s}^{n}} + \beta \mathcal{G}(\mathbf{s};\Theta)$}
        \EndFor
    \end{algorithmic}
\end{algorithm}

\subsection{Training and Evaluation}

Consider training a recurrent network, recurrent networks can be trained with a full sequence of events at once or it can be trained for each timestep. For clarity, in this paper we will look at the training at each timestep, $n$ and all simulations $m$.

\textbf{Case: $\mathbf{n = N}$.}
At the final timestep $N$, we let $y^N(\mathbf{S}^N) = f^N(\mathbf{S}^N)$ and $\nabla y^N(\mathbf{S}^N) = \nabla f^N(\mathbf{S}^N)$. We consider a smoothed-payoff function given by using a scaled softplus activation function with a user defined parameter $\kappa$ given by
\begin{equation}
    y^N(\mathbf{S}^N) = f^N_\kappa(\mathbf{S}^N) \equiv \frac{1}{\kappa}\ln(1+e^{\kappa g(\mathbf{S}^N)}),
\label{eq:smoothing}
\end{equation}
where $g(\mathbf{S}_m^N)$ is the same as in \eqref{eq:payoff}, let $\nabla g(\mathbf{S}^N) = [\frac{\partial g(\mathbf{S}^N)}{\partial \vv{S}_{1}^N},...,\frac{\partial g(\mathbf{S}^N)}{\partial \vv{S}_{d}^N}]$, then the derivative of the smoothed payoff is
\begin{equation}
    \nabla y^N(\mathbf{S}^N) = \nabla f^N_\kappa(\mathbf{S}^N) \equiv \hat{\sigma}(\kappa g(\mathbf{S}^N)) \nabla g(\mathbf{S}^N).
\end{equation} 

The hidden states are initialized as $f_\kappa^N(\mathbf{S}^N)$ for $\mathcal{F}(\mathbf{S}^N;\Omega)$ and $\nabla f_{\kappa}^N(\mathbf{S}^N)$ for $\mathcal{G}(\mathbf{S}^N;\Theta)$. The initial value of the continuation price is given by $c^N(\mathbf{S}_m^N) = f_\kappa^N(\mathbf{S}_m^N)$, which is the payoff at $t=T$. Similarly, the initial delta is given by $\nabla c^N(\mathbf{S}_m^N) = \nabla f_\kappa^N(\mathbf{S}_m^N)$. The initial hidden states are given by
\begin{equation*}
    \mathbf{h}^N = [f_{\kappa}^{N}(\mathbf{S}^N)],
    \label{eq:hidden_net_price}
\end{equation*}
for the pricing network and
\begin{equation*}
    \mathbf{h}^N = [\nabla f_{\kappa}^{N}(\mathbf{S}^N)],
    \label{eq:hidden_net_delta}
\end{equation*}
for the delta network.

\textbf{Case: $\mathbf{n=N-1,...,0}$.}
We use the deep RNN framework to compute the continuation price and delta of the option at timesteps $n=N-1,...,0$. We construct the training input as
\begin{equation*}
    \mathbf{X}^n = [\mathbf{S}^n,g(\mathbf{S}^n)]^T,
    \label{eq:input_net}.
\end{equation*}
The hidden state are updated as in standard recurrent networks, to carry the information to the next time step. Finally, we update the continuation price and delta of the option. For the optimal stopping time $\tilde{n}$, the continuation price at $n$ is given by
\begin{equation*}
    c^n= e^{-r(\tilde{n} - n)\Delta t}f^{\tilde{n}}
\end{equation*}
Similarly, we want to calculate the delta and the delta is updated using
\begin{equation*}
    \nabla c^n= e^{-r(\tilde{n} - n)\Delta t}\nabla f^{\tilde{n}}(\mathbf{S}^{\tilde{n}})\frac{\mathbf{S}^{\tilde{n}}}{\mathbf{S}^{n}}
\end{equation*}
We train the network in Algorithm \ref{alg:price_network} to approximate $y(\mathbf{S};\Omega^*)$ and the network in Algorithm \ref{alg:delta_network} to approximate $\nabla y(\mathbf{S};\Theta^*)$. To obtain the optimal set of parameter $\{\Omega^*,\Theta^*\}$ for each $n$, we optimize the networks over the loss function \eqref{eq:loss}, we can rewrite the loss function as a function of $\{\Omega,\Theta\}$ as
\begin{align}
    \label{eq:loss_net}
    \mathcal{L}(\Omega,\Theta)&=\mathbb{E}[|\mathbf{c}^{n}-y(\mathbf{S};\Omega)|^2] \\
    &+ \Delta t \mathbb{E}[|  \tilde\sigma \nabla \mathbf{c}^{n} - \tilde\sigma\nabla y(\mathbf{S};\Theta)|^2]. \nonumber
\end{align}
Minimizing the loss functions gives us the optimal set of parameters, we can express this as
\begin{equation}
    \{\Omega^*,\Theta^*\} = \textnormal{arg}\min\limits_{\{\Omega,\Theta\}}\{\mathcal{L}(\Omega,\Theta)\}.
    \label{eq:weight_opt}
\end{equation}
To solve \eqref{eq:weight_opt} an optimization algorithm like Adam \cite{kingma-2014} is used. The optimal parameter set naturally gives us the approximate solutions, $\mathbf{y}^{n*} = y^n(\mathbf{S}^n;\Omega^*)$, and $\nabla \mathbf{y}^{n*}= \nabla y^n(\mathbf{S}^n;\Theta^*)$ which solves problem \eqref{eq:optprob}.\\
The full algorithm for training is presented in Algorithm \ref{alg:training}, and we let $\mathcal{O}$ be the optimizer that minimizes \eqref{eq:loss} and we let $\mathbf{h}_\Delta^n$ be the hidden state of the delta approximating network. The $\max$ function is the element wise maximum between its inputs. To simplify some notation we let $\mathbf{y}_m^n \equiv y^n(\mathbf{S}^n;\Omega)$, and let $\nabla \mathbf{y}_m^n \equiv \nabla y^n(\mathbf{S}^n;\Theta)$. 
\begin{algorithm}[!htbp]
    \caption{Training the recurrent network}
    \label{alg:training}
    \begin{algorithmic}
        \Require{$\mathbf{S}^n,\mathcal{O}$}
        \State{initialize $\mathbf{c}$,$\nabla \mathbf{c}$,$\mathbf{y}$, and $\nabla \mathbf{y}$ of shape ($N+1$,$d$,batch size)}
        \For{$\mathbf{S}^n\textnormal{ in Batch}$}
            \State{$n = N$}
            \Let{$\mathbf{c}^N$,$\nabla \mathbf{c}^N$}{$f^N(\mathbf{S}^N)$,$\nabla f^N(\mathbf{S}^N)$}
            \Let{$\mathbf{y}^N$,$\nabla \mathbf{y}^N$}{$f^N(\mathbf{S}^N)$,$\nabla f^N(\mathbf{S}^N)$}
            \Let{$\{\mathbf{h}^N,\mathbf{h}_{\Delta}^N\}$}{$\{f^{N}(\mathbf{S}^N),\nabla f^{N}(\mathbf{S}^N)\}$}
            \For{$n=N-1,...,0$}
                \Let{$\mathbf{X}^n$}{$[\mathbf{S}^n,g(\mathbf{S}^n)]^\top$}
                \Let{$\tilde{n}$}{$\max\limits_{\tilde{n}\in[N-1,...,n+1]}\mathbb{E}[f(\mathbf{S})]$}
                \Let{$\mathbf{c}^n$,$\nabla\mathbf{c}^n$}{$e^{-r(\tilde{n}-n)\Delta t}f^{\tilde{n}}(\mathbf{S}^{\tilde{n}})$,$e^{-r(\tilde{n}-n)\Delta t}\nabla f^{\tilde{n}}(\mathbf{S}^{\tilde{n}})\frac{S^{\tilde{n}}}{S^{n}}$}
            \EndFor
            \Let{$\mathbf{y}^{N-1,...,0}$}{Algorithm \ref{alg:price_network}}
            \Let{$\nabla \mathbf{y}^{N-1,...,0}$}{Algorithm \ref{alg:delta_network}}
            \Let{$\{\mathbf{y}^{*},\nabla \mathbf{y}^{*}\}$}{$\mathcal{O}(\mathcal{L}(\mathbf{y},\mathbf{c},\nabla \mathbf{y},\nabla \mathbf{c}))$}
        \EndFor
    \end{algorithmic}
\end{algorithm}

Once training is complete, we evaluate the trained network, with a forward pass of the network. Similar to the training phase, the recurrent networks can evaluate the full sequence of events at once. The initial setup at time $n=N$ is the same and will not be repeated.

In the evaluation phase, we directly construct the evaluation input \eqref{eq:input_net}. With the hidden state at $n=N$ we use the trained network Algorithm \ref{alg:price_network} to approximate $y^n(\mathbf{S}^n;\Omega^*)$ and we use the trained network Algorithm \ref{alg:delta_network} to approximate $\nabla y^n(\mathbf{S}^n;\Theta^*)$ for $n=N-1,...,0$. For the evaluation phase, we use the learned price, $y^n(\mathbf{S}^n;\Omega^*)$ to construct the exercise boundary which is given by
\[
    \mathcal{E}^{n*} = 
    \begin{cases}
        1\ &\ \textnormal{if}\ f_{\kappa}^n(\mathbf{S}^n)\geq y^n(\mathbf{S}^n;\Omega^*)\\
        0\ &\ \textnormal{otherwise}.
    \end{cases}
\]
Then the value update is given by
\begin{equation}
    \mathbf{v}^n \approx \max\{f^{n}(S^{n}),y^{n}(\mathbf{S}^{n};\Omega^*)\}
    \label{eq:price_eval}
\end{equation}
and the delta is updated as
\begin{equation}
    \nabla \mathbf{v}^n \approx \nabla f^{n}(\mathbf{S}^n)\mathcal{E}^{n*} + \nabla y^{n}(\mathbf{S}^n;\Theta^*)(1-\mathcal{E}^{n*})
\end{equation}
The full algorithm for evaluation is presented in Algorithm \ref{alg:eval}.
\begin{algorithm}[!htbp]
    \caption{Evaluation of Recurrent network}
    \label{alg:eval}
    \begin{algorithmic}
        \Require{$\mathbf{S}^n,\epsilon$}
        \State{initialize $\mathbf{v}$,$\nabla \mathbf{v}$,$\mathbf{c}$,$\nabla \mathbf{c}$,$\mathbf{y}$, and $\nabla \mathbf{y}$ of shape ($N+1$,$d$,batch size)}
        \For{$\mathbf{S}^n\textnormal{ in Batch}$}
            \State{$n = N$}
            \Let{$\mathbf{v}^N, \nabla \mathbf{v}^N$}{$f^N(\mathbf{S}^N), \nabla f^N(\mathbf{S}^N)$}
            \Let{$\mathbf{c}^N$,$\nabla \mathbf{c}^N$}{$f^N(\mathbf{S}^N)$,$\nabla f^N(\mathbf{S}^N)$}
            \Let{$\mathbf{y}^N$,$\nabla \mathbf{y}^N$}{$f^N(\mathbf{S}^N)$,$\nabla f^N(\mathbf{S}^N)$}
            \Let{$\{\mathbf{h}^N,\mathbf{h}_{\Delta}^N\}$}{$\{f^{N}(\mathbf{S}^N),\nabla f^{N}(\mathbf{S}^N)\}$}
            \For{$n=N-1,...,0$}
                \Let{$\mathbf{X}^n$}{$[\mathbf{S}^n,g(\mathbf{S}^n)]^\top$}
                \Let{$\tilde{n}$}{$\max\limits_{\tilde{n}\in[N-1,...,n+1]}\mathbb{E}[f(\mathbf{S})]$}
                \Let{$\mathbf{c}^n$,$\nabla\mathbf{c}^n$}{$e^{-r(\tilde{n}-n)\Delta t}f^{\tilde{n}}(\mathbf{S}^{\tilde{n}})$,$e^{-r(\tilde{n}-n)\Delta t}\nabla f^{\tilde{n}}(\mathbf{S}^{\tilde{n}})\frac{S^{\tilde{n}}}{S^{n}}$}
            \EndFor
            \Let{$\mathbf{y}^{*[N-1,...,0]}$}{Algorithm \ref{alg:price_network}}
            \Let{$\nabla \mathbf{y}^{*[N-1,...,0]}$}{Algorithm \ref{alg:delta_network}}
            \Let{$\mathcal{E}^{*}$}{$(f(\mathbf{S}) \geq \mathbf{y}^{*})$}
           \Let{$\mathbf{v}$}{$\max\{f(\mathbf{S}),\mathbf{y}^{*}\}$}
           \Let{$\nabla \mathbf{v}$}{$\nabla f(\mathbf{S})\mathcal{E}^{*}  + \nabla \mathbf{y}^{*} (1-\mathcal{E}^{*})$}
        \EndFor
    \end{algorithmic}
\end{algorithm}

\section{Computational Costs}\label{sec:computation}

In this section, we analyze the computational cost of our proposed algorithm, and make a comparison with the DRL method. The DRL method and other feedforward methods require the neural networks to be initialized and trained at each time step. The method in \cite{chen-2019} attempts to reduce the RAM burden by pricing at each time step and storing the results. However, this means that they must write to a hard-drive at each time step which introduces additional time inefficiencies and large overhead for memory.

The complexity of training and pricing for the DRL method has been shown to be $\mathcal{O}((\frac{c_1N}{J}+c_2)NMLd^2)$ \cite{chen-2019}, where $N$ is the number of time steps, $M$ is the number of simulations, $L$ is the number of neural network layers, $d$ is the dimensions, $J$ is the number of timestep skips, and $\{c_1,c_2\}$ are constants related to the batch size and number of epochs. Note that $J$ is a small constant, typically around $2-5$. Thus the complexity is quadratic in $N$.

The memory complexity is calculated to be $\mathcal{O}(NMd)$ \cite{chen-2019}, for storing the asset prices, price approximation $y^n$ and the delta approximation $\nabla y^{n+1}$. However, this does not count the number of network weights and trainable parameters that must be computed and stored at each timestep. Considering all of the network parameters, the total memory complexity amounts to $\mathcal{O}(NMLd^2)$.

Our proposed method allows for time and memory savings as we only need to train one set of weights rather than $N/J$ sets of weights. In any neural network, the training dominates both the time and memory requirements. In our method, Algorithm \ref{alg:training} dominates the total computational time. This is because the training involves the optimization of loss function \eqref{eq:loss_net} at each step.

\subsection{Time Cost}

We will separate the analysis of two cases; one at the final time and one for the other time steps.

\textbf{Case 1: $\mathbf{n=N}$.}
At final time when $n=N$, the proposed method needs to compute $\mathcal{E}^N$, $v^N$ and $\nabla v^N$, each of which is an $M\times d$ matrix. This is performed twice, once for Algorithm \ref{alg:training} and once for Algorithm \ref{alg:eval}. For each Algorithm \ref{alg:training} and Algorithm \ref{alg:eval}, the time complexity at $N$ for $\mathcal{E}^N$ is constant. The evaluation of $f^N(S^N)$ and $\nabla f^N(S^N)$ over an array of size $M \times 2d$ both require $\mathcal{O}(Md)$ operations. Thus the time complexity requirement at $N$ is $\mathcal{O}(Md)$.

\textbf{Case 2: $\mathbf{n=N-1,...,0}$.}
At timesteps $n=N-1,...,0$, we need to construct the input $X^n$ an $M\times 2d$ matrix, initialize and optimize the Recurrent networks $y^n$ and $\nabla y^n$. The construction of the input $X^n$ and the initialization of the Recurrent network is done in constant time. Thus we will not focus on these steps. The majority of the operations required in Algorithm \ref{alg:training} is during optimization. For input size $M \times 2d$ we have five weights of size $2d \times 2d$ and an output weight of size $2d \times d$ that need to be optimized over $L$ layers. The dominant operation during optimization is multiplication of weights and inputs which requires $\mathcal{O}(MLd^2)$ operations. Thus the resulting complexity for training the Recurrent network over $N$ timesteps is $\mathcal{O}(NMLd^2)$. The time complexity for Algorithm \ref{alg:eval} is run in a forward manner through $N$ timesteps and $M$ samples. Since the forward operation is multiplicative, through each layer, we get a complexity of $\mathcal{O}(NMLd^2)$ for the price and delta approximations. Thus, the total time complexity for Algorithms \ref{alg:training} and \ref{alg:eval} is bounded by $\mathcal{O}(NMLd^2)$.

Using the python library time we can track the time requirements of the DRL method and compare it to our method, the total time used by both methods are presented in Table \ref{tab:time_memory}. The complexity of the methods were measured by taking the logarithm of timing results from Table \ref{tab:time_memory} and presented in Figure \ref{fig:time_N}. Figure \ref{fig:time_N} shows that the DRL method has a worse than linear growth with respect to $N$ and near linear growth of time in our method. As shown in Table \ref{tab:time_memory}, we can clearly see the absolute advantage we achieve by reducing the complexity by a factor of $N$, our method even at $N=100$ is already $10$ times faster than the DRL method. Also the time results did not vary significantly between dimension $2,5$ and $10$, but showed more variation when the number of timesteps was increased.

\begin{table}[!htbp]
\caption{Time and Memory comparison between methods. For this experiment the American style geometric average call option with $d=2, 5$ and $10$, priced using $M=50000$. We used a Recurrent layer size of $L=7$ for our method. The network skip number $J = 4$ and $L=7$ for the method in \cite{chen-2019}.}
\begin{subtable}{\textwidth}
    \centering
    \begin{tabular}{ccccc}
        Time steps & \multicolumn{4}{c}{Proposed method}\\
        \hline \hline
        N & Total Time (s) & Training Time (s) & Memory (MB) & Price \\ \hline
        50 & 39.48 & 29.2 & 5.62 & 6.6601\\
        100 & 76.35 & 56.73 & 5.63 & 6.4266\\
        1000 & 516.07 & 379.94 & 5.63 & 6.5961\\ \hline \\
    \end{tabular}
    \begin{tabular}{ccccc}
        Time steps & \multicolumn{4}{c}{DRL method}\\
        \hline \hline
        N & Total Time (s) & Training Time (s) & Memory (MB) & Price \\ \hline
        50 & 113.75 & 113.62 & 3701.59 & 6.6702\\
        100 & 316.3 & 316.06 & 7705.53 & 6.4817\\
        1000 & 25590.97 & 25586.84 & 151944.60 & 6.5822\\ \hline
    \end{tabular}
    \caption{d=2}
    \begin{tabular}{ccccc}
        Time steps & \multicolumn{4}{c}{Proposed method}\\
        \hline \hline
        N & Total Time (s) & Training Time (s) & Memory (MB) & Price \\ \hline
        50 & 39.32 & 28.91 & 5.62 & 5.9069\\
        100 & 78.82 & 57.77 & 5.63 & 5.9524\\
        1000 & 812.96 & 577.02 & 5.63 & 5.8701\\ \hline \\
    \end{tabular}
    \begin{tabular}{ccccc}
        Time steps & \multicolumn{4}{c}{DRL method}\\
        \hline \hline
        N & Total Time (s) & Training Time (s) & Memory (MB) & Price \\ \hline
        50 & 134.8 & 134.64 & 4548.33 & 5.9218\\
        100 & 410.92 & 410.62 & 9408.30 & 5.9518\\
        1000 & 36739.60 & 36734.08 & 168054.00 & 5.8697\\ \hline
    \end{tabular}
    \caption{d=5}
    \begin{tabular}{ccccc}
        Time steps & \multicolumn{4}{c}{Proposed method}\\
        \hline \hline
        N & Total Time (s) & Training Time (s) & Memory (MB) & Price \\ \hline
        50 & 38.55 & 28.15 & 5.62 & 5.9883\\
        100 & 76.44 & 56.05 & 5.63 & 5.9087\\
        1000 & 785.49 & 552.97 & 5.63 & 5.7964\\ \hline \\
    \end{tabular}
    \begin{tabular}{ccccc}
        Time steps & \multicolumn{4}{c}{DRL method}\\
        \hline \hline
        N & Total Time (s) & Training Time (s) & Memory (MB) & Price \\ \hline
        50 & 251.40 & 251.02 & 6847.64 & 5.9791\\
        100 & 726.46 & 725.86 & 14124.60 & 5.9053\\
        1000 & 65499.84 & 65491.55 & 215701.40 & 5.7364\\ \hline
    \end{tabular}
    \caption{d=10}
\end{subtable}
\label{tab:time_memory}
\end{table}

\begin{figure}[!htbp]
    \centering
    \begin{subfigure}{0.3\textwidth}
         \centering
         \includegraphics[width=\textwidth]{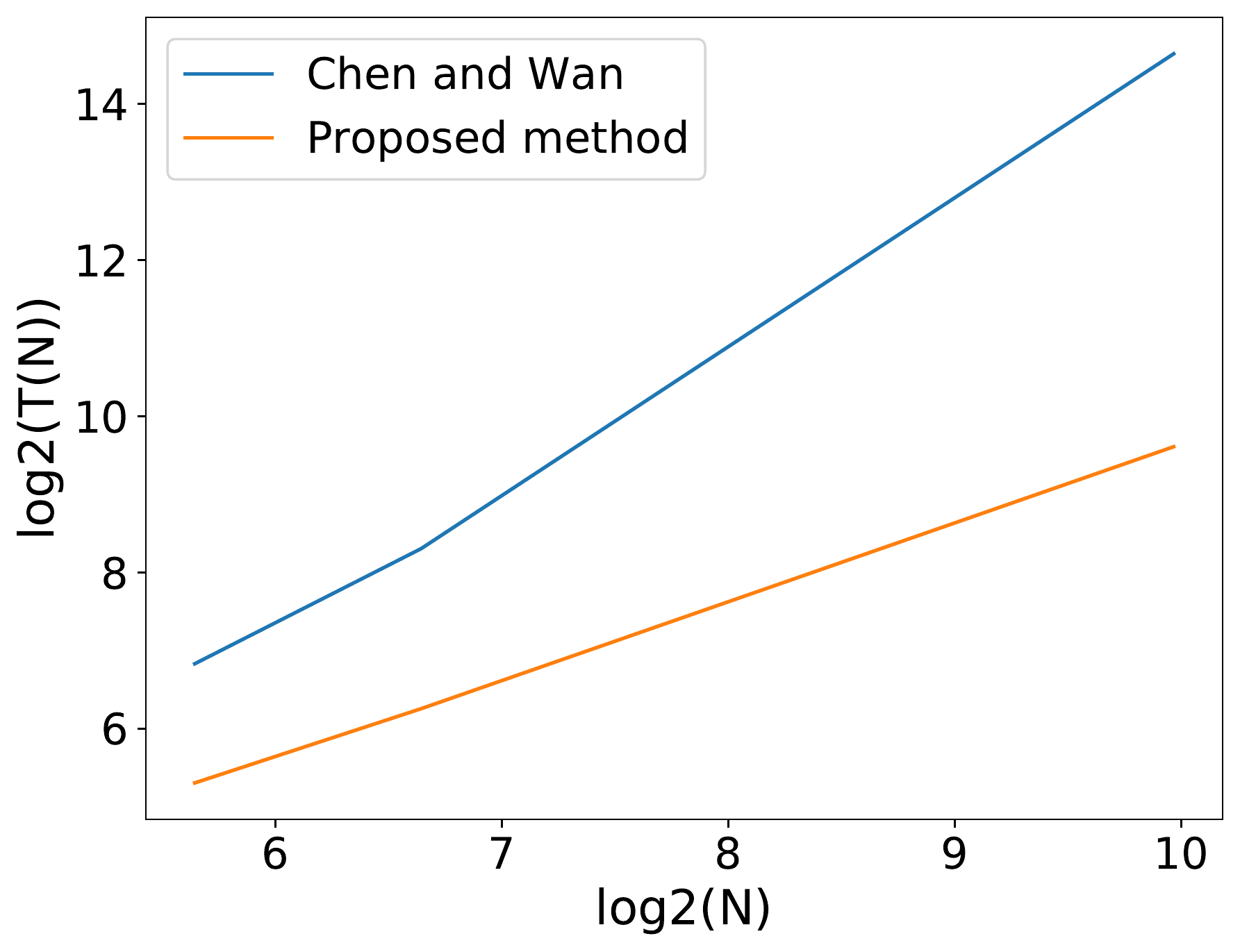}
         \caption{$d=2$}
         \label{fig:time_d2}
    \end{subfigure}
    \hfill
    \begin{subfigure}{0.3\textwidth}
         \centering
         \includegraphics[width=\textwidth]{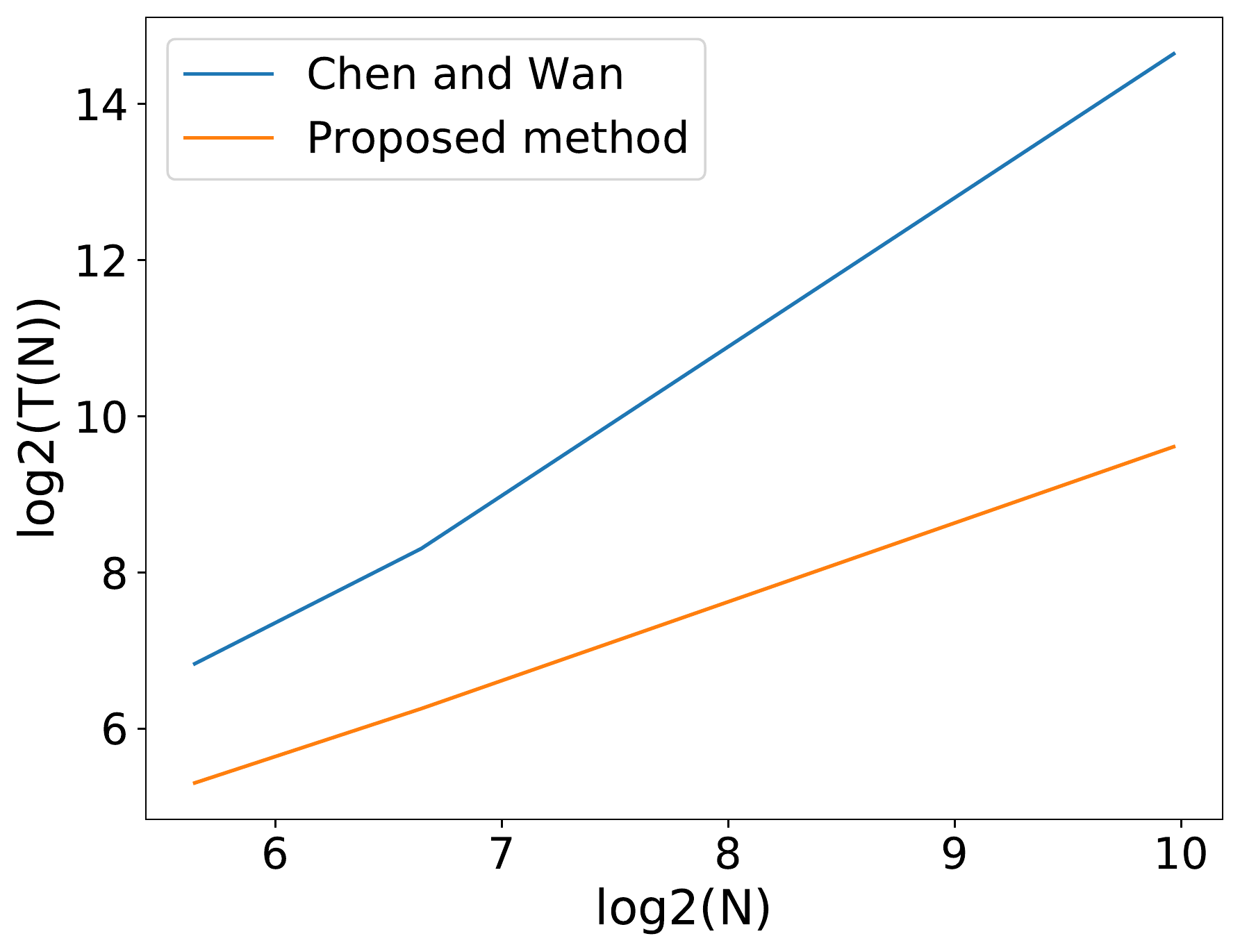}
         \caption{$d=5$}
         \label{fig:time_d5}
    \end{subfigure}
    \hfill
    \begin{subfigure}{0.3\textwidth}
         \centering
         \includegraphics[width=\textwidth]{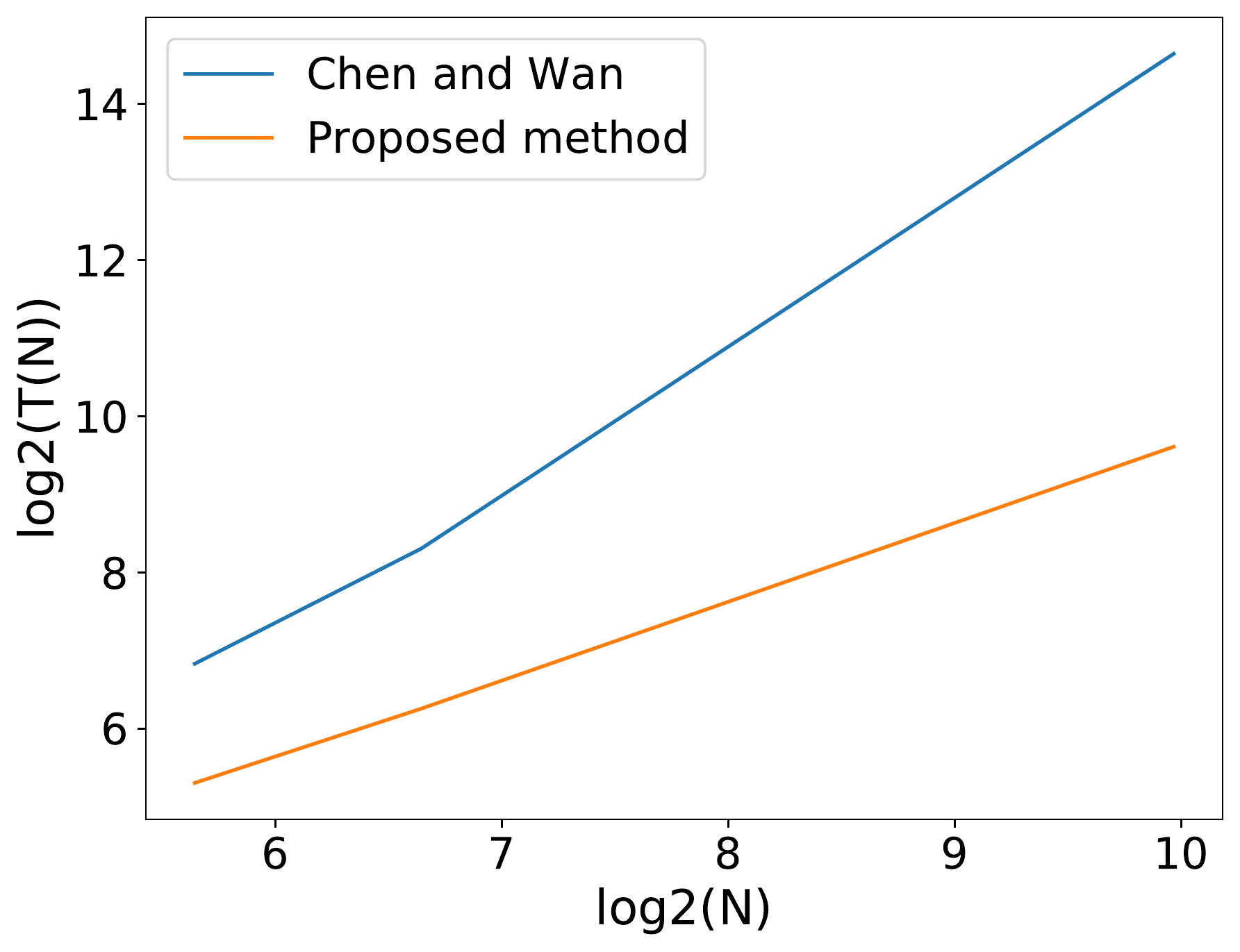}
         \caption{$d=10$}
         \label{fig:time_d10}
    \end{subfigure}
    \caption{Logarithm of the total time, T(N), against the number of time steps, $N$, for the DRL method (blue line) and the proposed method (orange line) for dimensions (a) $d=2$, (b) $d=5$ and (c) $d=10$.}
    \label{fig:time_N}
\end{figure}

\subsection{Memory Cost}

Another contribution of our work is the reduction of memory complexity by a factor of $N$ due to the use of a Recurrent network. This is a large reduction as $N$ is the dominant number in our complexity, since a large $N$ is required for more accurate solutions. 

The initialized asset prices, option payoff $f^n(S^n)$ and delta $\nabla f^n(S^n)$ requires $\mathcal{O}(NMd)$ floating points of storage, which is the same as the DRL method \cite{chen-2019}. Our methods differs in terms of the network, thus we only need to compare the storage requirements of the network. Since, the proposed method only stores the training parameters at the final timestep, it is not required to store training weights and gradient information at each timestep. Thus for the input of size $M \times 2d$, hidden layers of size $2Ld \times 2d$ and output layer of size $2d \times d$, the Recurrent networks requires $\mathcal{O}(MLd^2)$ floating points, which is constant in $N$.

As for comparison, using the python library tracemalloc we measured the total peak memory requirements of the DRL method and compared it to our method. This was done by measuring the peak memory of each function and summing the accumulated peak memories. The total peak memory used by both methods are presented in Table \ref{tab:time_memory}. The memory complexity of the methods were measured by taking the logarithm of the peak memory measured for a given number of timesteps $N$ and presented in Figure \ref{fig:memory_N}. Figure \ref{fig:memory_N} shows the near linear growth of memory in $N$ for the DRL method and a constant growth in our method. Also the memory results did not vary significantly between dimension $2,5$ and $10$, but showed more variation when the number of timesteps was increased.

\begin{figure}[!htbp]
    \centering
    \begin{subfigure}{0.3\textwidth}
         \centering
         \includegraphics[width=\textwidth]{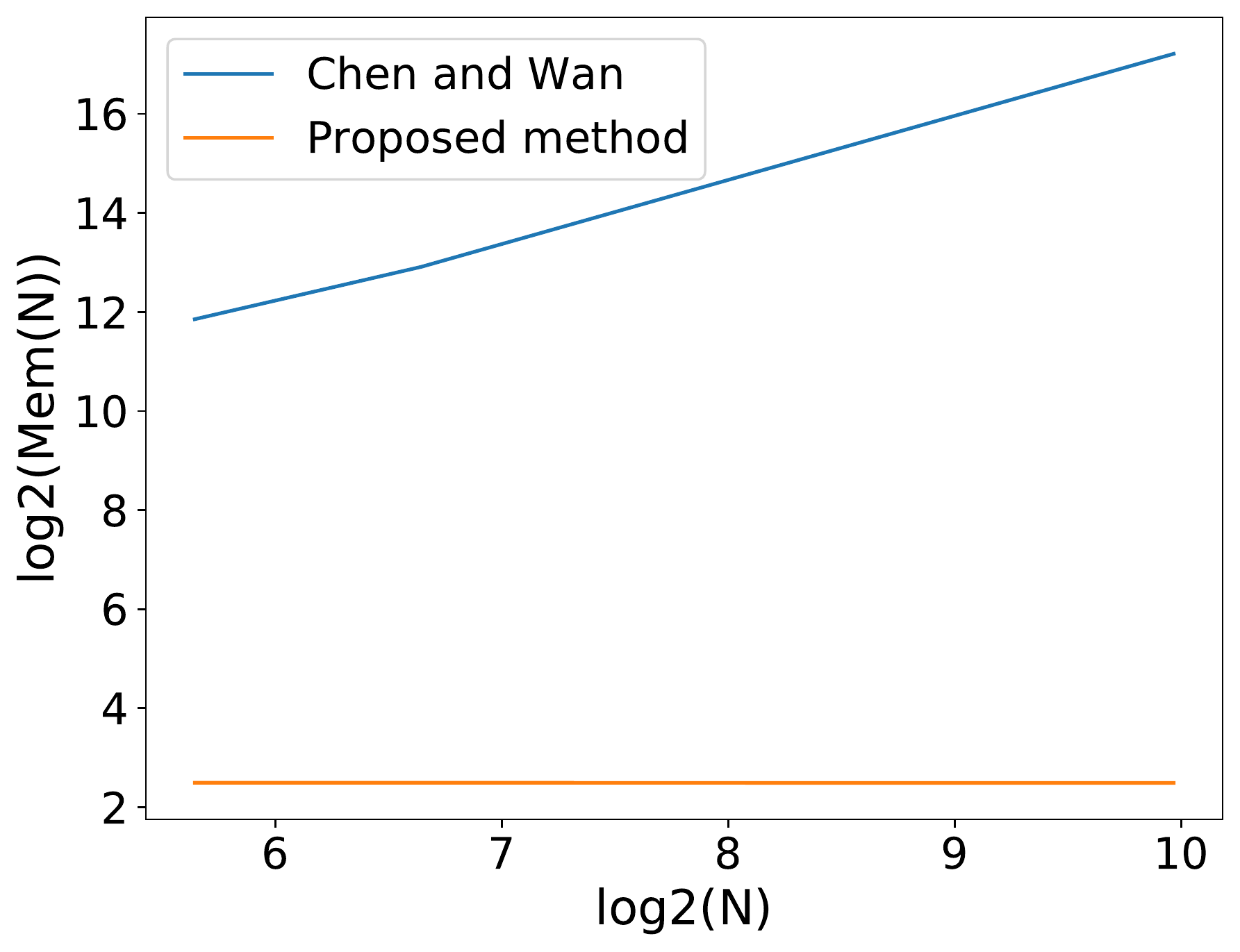}
         \caption{$d=2$}
         \label{fig:time_d2}
    \end{subfigure}
    \hfill
    \begin{subfigure}{0.3\textwidth}
         \centering
         \includegraphics[width=\textwidth]{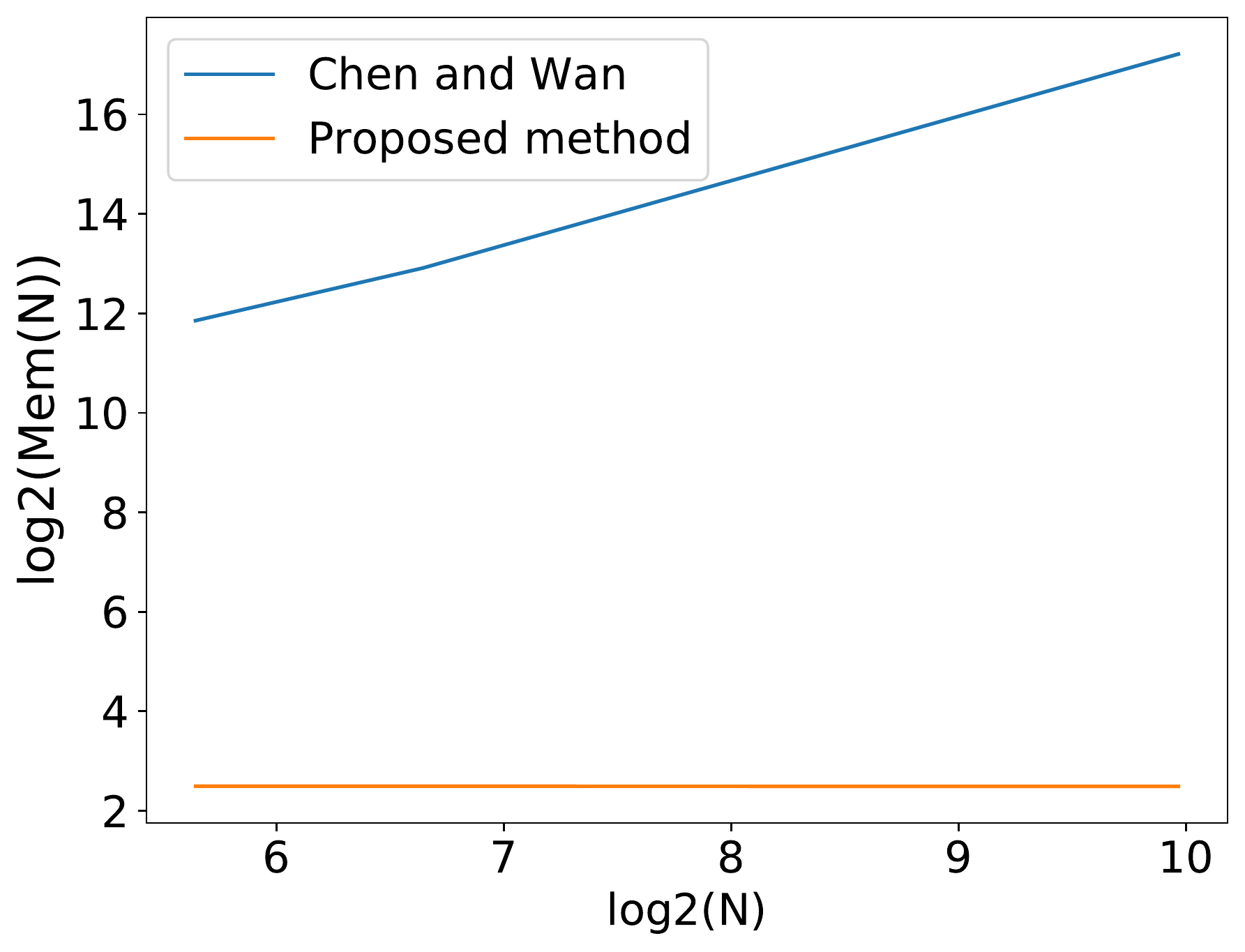}
         \caption{$d=5$}
         \label{fig:time_d5}
    \end{subfigure}
    \hfill
    \begin{subfigure}{0.3\textwidth}
         \centering
         \includegraphics[width=\textwidth]{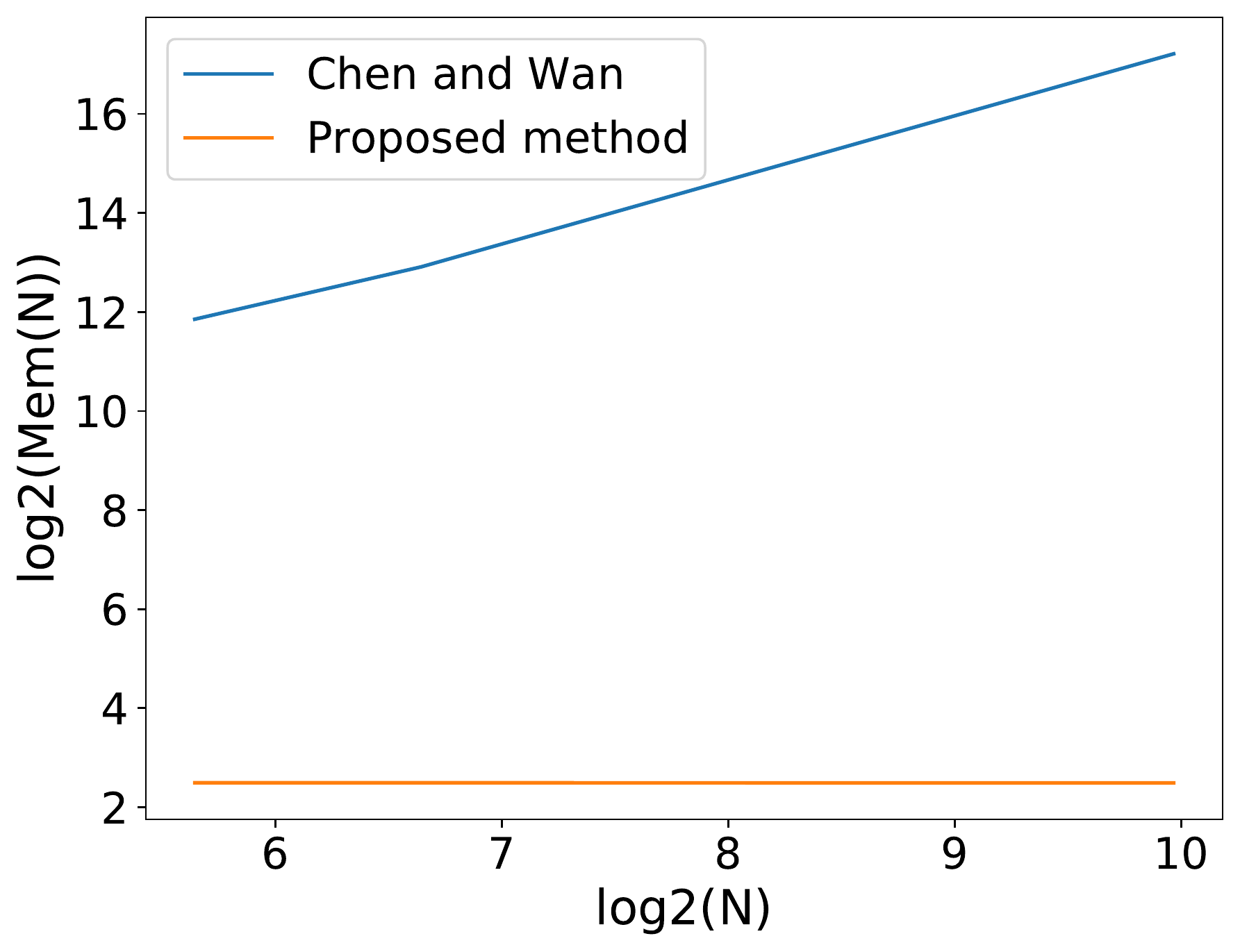}
         \caption{$d=10$}
         \label{fig:time_d10}
    \end{subfigure}
    \caption{Logarithm of the total peak memory against the number of time steps, $N$, for the DRL method (blue line) and the proposed method (orange line) for dimensions (a) $d=2$, (b) $d=5$ and (c) $d=10$.}
    \label{fig:memory_N}
\end{figure}

\section{Numerical Results}\label{sec:Experiments}

In this section, we solve the high-dimensional American option problem using our Algorithms \ref{alg:training} and \ref{alg:eval}. We compute the price $v(\vv{s}_0,0)$ and the delta $\nabla v(\vv{s}_0,0)$ at $t = 0$ for given $\vv{s}_0 = (s_{0}^1,\dots,s_{0}^d)$ where $s_{0}^1 = \dots = s_0^d = 0.9K,K$ or $1.1K$. We also show that our method can be used for delta hedging across spacetime, by demonstrating our hedging positions in a $2$-dimensional case. We use two separate types of experiments. In the first type we constrain all computational runtime such that the runtime for our proposed method and the runtime for the DRL method are the same. In the second type of experiment we run the experiment until the pricing errors produced by our method matches the pricing errors of the DRL method.

In our experiments we set the strike $K=100$. The smoothing parameter in \eqref{eq:smoothing}, is given by $\kappa = 2/\Delta t$. The network used in our Experiments is a $L=7$ deep RNN with GRU units. The batch size for our proposed method was given as $100000$. Our computation was performed using a 6GB-NVIDIA GTX 2060 GPU, a 6 core AMD Ryzen 5 3600X processor and 16GB of RAM.

In Experiment $1$, we look at the American style geometric average option. We compare the accuracy of the proposed method to the DRL method for both price and delta at time $t=0$ for $S_0 = s_0$ for all simulations, while under similar runtime constraint. In Experiment $2$, we compare the delta hedging results calculated using the price and delta computed by the proposed method and the DRL methods, we present the portfolio mean and standard deviations for dimensions $[2,5,30,100]$. In the Experiments $3$, we compare the price at $t=0$, delta at $t=0$ and delta hedging results for the American style max call option under the similar runtime constraints. Experiment $4$ and $5$ evaluates how long the proposed method needs to run in order to match the accuracy of the DRL method. For Experiments $1$, $2$ and $3$, finite difference solutions with very fine grids are used as exact solutions. We note that this is feasible only if $d \leq 3$. We make comparison results to the unsupervised method of \cite{salvador-2020}, and also look at the far out-the-money(OTM) and far in-the-money(ITM) performance of our method, the DRL method and the Longstaff-Schwartz method in Experiment $3$. In many cases we do not make a comparison to the Longstaff-Schwartz method as it has been done in \cite{chen-2019}. In addition, we remark that the comparison is not made with the other methods referenced in the introduction, such as \cite{guler-2019}, \cite{salvador-2020}, \cite{sirignano-2018}, \cite{e-2017}, \cite{beck-2018} and \cite{han-2016}. This is because the work of \cite{sirignano-2018} was compared in \cite{chen-2019}. Even though the works of \cite{e-2017}, \cite{beck-2018} and \cite{han-2016} compute European option problems using neural networks based on BSDEs, these works do not discuss the more challenging American option problems. The work of \cite{guler-2019} looks at solving the Black-Scholes and HJB equations using a BSDE approach however they look at solving the general PDE and not the American option problem. Their approach is also simulated with very limited number of timesteps, simulation size and the dimensions are limited to $d = 2$. 

We note that when finite difference solutions are available, we can evaluate the absolute and percent errors of computed prices and deltas. More specifically, denote the finite difference solutions as $v_{\text{exact}}$. Then the percent errors of the price at t = 0 is
\begin{equation}
    \frac{|v(\vv{s}_0,0) - v_{\text{exact}}|}{|v_{\text{exact}}|} \times 100\%,
    \label{eq:price_rel_err}
\end{equation}
and the percent errors of the delta at t=0 is
\begin{equation}
    \frac{||\nabla v(\vv{s}_0,0) - \nabla v_{\text{exact}}||_{L_2}}{||\nabla v_{\text{exact}}||_{L_2}} \times 100\%.
    \label{eq:delta_rel_error}
\end{equation}

In addition, we evaluate the quality of computed exercise boundaries. In the first set of experiments, we look at the exercise boundary, i.e. a point $(\vv{S}^n,t^n)$ is 'exercised' or 'continued', computed by our proposed method and the DRL method, under run time constraint. The true exercise points is determined by finite difference method without any run time constraints. Let ‘exercised’ class be the positive class, and denote the numbers of true positive, true negative, false positive and false negative samples as TP, TN, FP, FN, respectively. Then the quality of the exercise boundaries can be evaluated as
\begin{equation}
    \text{f1-score} = \frac{TP}{TP+0.5(FP+FN)}.
\label{eq:f1}
\end{equation}
The best (or worst) case of the f1-score is 1 (or 0), respectively. Another common metric to evaluate the quality of classification problems is the accuracy, however, since all of our experiments are skewed to the positive class, the f1-score would be a better metric than the accuracy \cite{chen-2019}.

\subsection{American Style Geometric Average Call Option}

We consider a $d$-dimensional American style ‘geometric average’ call option, where $\rho_{ij} = \rho$ for $i\neq j$, $\sigma_i = \sigma$ for all $i$’s, and the payoff function is given by $f(\vv{s}) = \max\{(\prod_{i=1}^{d} s_i)^{1/d} - K, 0\}$. Although, geometric average options are rarely seen in practice, they have semi-analytical solutions for benchmarking the performance of our algorithm in high dimensions, as shown in \cite{glasserman-2004}, \cite{sirignano-2018} and \cite{chen-2019}. More specifically, it is shown in \cite{glasserman-2004}, that the $d$-dimensional can be reduced to a one dimensional American call option with the variable $s' = (\prod_{i=1}^{d}s_i)^{1/d}$, where the effective volatility is $\sigma' =\sqrt{(1+(d-1))/d\sigma}$, and the effective drift, $\mu' = r-\delta +\frac{1}{2}(\sigma'^2-\sigma^2)$. Thus, by solving the equivalent one-dimensional option using finite difference method, one can compute the $d$-dimensional option prices and under special circumstances, deltas\footnote{We note that solving the equivalent one-dimensional option is not sufficient for computing the $d$-dimensional delta except at the symmetric points $s_1 = \dots = s_d$ . Interested readers can verify this by straightforward algebra.} accurately.

In the following Experiments $1$ and $2$, we evaluate the American style geometric average call option under a run time constraints. For each $s_0=90,100,110$ and $d=2,5,30,100$ we run our method and the DRL method for a fixed run time in seconds. Through this experiment we show that under time constraints, our method approximates the solution more accurately than the DRL method. We computed price of the American style geometric average call option as described in \cite{sirignano-2018}, with the parameters $\{T=2,r=0,\sigma=0.25,\rho=0.75\}$.

\subsubsection*{Experiment 1: American style geometric average call option}

In the first experiment we compare the computed prices at $t = 0$ between our proposed method and the deep residual network of Recurrent under a run time constraint; see Tables \ref{tab:2-d_geo} - \ref{tab:100-d_geo}. Each Table includes: the exact prices computed by the Crank-Nicolson finite difference method with $1000$ timesteps and $16,385$ space grid points, the prices, corresponding percent errors computed and the run time of the DRL method, the prices, the percent errors computed and run time of our proposed method. For the proposed method, the percent errors range from $0.29\% - 2.29\%$ for computed prices. Our method performs noticeably better than the DRL method at lower dimensions as total sample sized used for $d=2,5,30,100$ was $200000$. As a comparison, for the DRL method, the percent errors range from $2.04\% - 21.01\%$. To ensure the computation did not exceed the runtime limit, the timestep was set to $N=10$ and the skip parameter $J=8$. This led to a serious degredation of accuracy for the DRL method. 

\begin{table}[!htbp]
    \centering
    \caption{2-dimensional geometric average call option: computed prices and deltas at $t=0$}
    \begin{tabular}{cccccc}
        & \multicolumn{5}{c}{Finite Difference method} \\ \hline \hline
        $s_0$ & Price & \% Error & Delta &  \% Error & Run Time (s) \\ \hline
        90 & 6.7000 & -- & 19.3136 & -- & -- \\
        100 & 11.2502 & -- & 25.8041 & -- & -- \\
        110 & 16.7708 & -- & 31.3886 & -- & -- \\ \hline \\
        & \multicolumn{5}{c}{Proposed method} \\ \hline \hline
        $s_0$ & Price & \% Error & Delta &  \% Error & Run Time (s) \\ \hline
        90 & 6.6520 & 0.71348 & 19.0308 & 1.46425 & 36.39 \\
        100 & 11.0801 &  1.93303 & 25.3053 & 1.93303 & 36.74 \\
        110 & 16.8600 &  0.5313 & 31.1031 & 0.90957 & 36.21 \\ \hline \\
        & \multicolumn{5}{c}{DRL method} \\ \hline \hline
        $s_0$ & Price & \% Error & Delta &  \% Error & Run Time (s) \\ \hline
        90 & 5.2920 &  21.0117 & 23.0260 &  19.2220 & 36.01 \\
        100 & 10.0960 & 18.2110 & 30.5032 & 18.2110 & 36.25 \\
        110 & 16.2260 & 3.24612 & 35.3606 & 12.6540 & 36.09\\ \hline
    \end{tabular}
    \label{tab:2-d_geo}
\end{table}
\begin{table}[!htbp]
    \centering
    \caption{5-dimensional geometric average call option: computed prices and deltas at $t=0$}
    \begin{tabular}{cccccc}
        & \multicolumn{5}{c}{Finite Difference method} \\ \hline \hline
        $s_0$ & Price & \% Error & Delta &  \% Error & Run Time (s) \\ \hline
        90 & 6.0685 & -- & 7.3985 & -- & -- \\
        100 & 10.3173 & -- & 10.0894 & -- & -- \\
        110 & 16.1580 & -- & 12.8505 & -- & -- \\ \hline \\
        & \multicolumn{5}{c}{Proposed method} \\ \hline \hline
        $s_0$ & Price & \% Error & Delta &  \% Error & Run Time (s) \\ \hline
        90 & 6.0820 & 0.2241 & 7.3345 & 0.86504 & 37.19 \\
        100 & 10.4070 & 0.8733 & 9.9467 & 1.41436 & 38.69 \\
        110 & 15.9818 & 1.08926 & 12.3168 & 4.153145792 & 37.54 \\ \hline \\
        & \multicolumn{5}{c}{DRL method} \\ \hline \hline
        $s_0$ & Price & \% Error & Delta &  \% Error & Run Time (s) \\ \hline
        90 & 5.6789 &  6.41345 & 7.9648 & 7.6543 & 37.06 \\
        100 & 9.1800 &  11.0281 & 12.6274 & 25.155 & 38.23 \\
        110 & 13.1910 &  18.3589 & 17.9543 & 32.713 & 37.07 \\ \hline
    \end{tabular}
    \label{tab:5-d_geo}
\end{table}
\begin{table}[!htbp]
    \centering
    \caption{30-dimensional geometric average call option: computed prices and deltas at $t=0$}
    \begin{tabular}{cccccc}
        & \multicolumn{5}{c}{Finite Difference method} \\ \hline \hline
        $s_0$ & Price & \% Error & Delta &  \% Error & Run Time (s) \\ \hline
        90 & 5.7740 & -- & 1.2028 & -- & -- \\
        100 & 9.9660 & -- & 1.6593 & -- & -- \\
        110 & 15.3390 & -- & 2.0697 & -- & -- \\ \hline \\
        & \multicolumn{5}{c}{Proposed method} \\ \hline \hline
        $s_0$ & Price & \% Error & Delta &  \% Error & Run Time (s) \\
        \hline
        90 & 5.7200 & 0.93701 & 1.1944 & 0.69837 & 52.49 \\
        100 & 9.8960 & 0.70136 & 1.6307 & 1.72362 & 52.37 \\
        110 & 15.4680 & 0.8449 & 2.0405 & 1.41083 & 53.26 \\ \hline \\
        & \multicolumn{5}{c}{DRL method} \\ \hline \hline
        $s_0$ & \% Error & Delta &  \% Error & Run Time (s) \\ \hline
        90 & 5.6557 &  2.04375 & 1.1907 &  1.00599 & 52.44 \\
        100 & 8.3430 & 16.2839 & 1.9863 & 19.707 & 52.44 \\
        110 & 14.8900 &  2.92532 & 2.3296 & 12.557 & 53.17 \\ \hline
    \end{tabular}
    \label{tab:30-d_geo}    
\end{table}
\begin{table}[!htbp]
    \centering
    \caption{100-dimensional geometric average call option: computed prices and delta at $t=0$}
    \begin{tabular}{cccccc}
        & \multicolumn{5}{c}{Finite Difference method} \\ \hline \hline
        $s_0$ & Price & \% Error & Delta &  \% Error & Run Time (s) \\ \hline
        90 & 5.6352 & -- & 0.3548 & -- & --\\
        100 & 10.0658 & -- & 0.4967 & -- & -- \\
        110 & 15.2953 & -- & 0.6272 & -- & --\\ \hline \\
        & \multicolumn{5}{c}{Proposed method} \\ \hline \hline
        $s_0$ & Price & \% Error & Delta &  \% Error & Run Time (s) \\ \hline
        90 & 5.6960 & 0.93701 & 0.3566 & 0.5073 & 212.74\\
        100 & 9.9118 &  1.52993 & 0.4903 & 1.2885 & 205.08\\
        110 & 15.3499 & 0.3570 & 0.6114 & 2.51913 & 206.21\\ \hline \\
        & \multicolumn{5}{c}{DRL method} \\ \hline \hline
        $s_0$ & Price & \% Error & Delta &  \% Error & Run Time (s) \\ \hline
        90 & 5.9140 & 4.9404 & 0.4046 &  14.036 & 212.46\\
        100 & 8.9989 & 10.5993 & 0.6195 & 24.723 & 205.59\\
        110 & 16.2755 & 6.4085 & 0.6789 &  8.243 & 206.43\\ \hline
    \end{tabular}
    \label{tab:100-d_geo}
\end{table}

By design, the DRL method gives us deltas in all spacetime. Tables \ref{tab:2-d_geo} - \ref{tab:100-d_geo}\footnote{Only a single value of delta is reported for the geometric average, since our experiment was run for $S_1 = \dots = S_d$ all deltas are equal as shown in \cite{chen-2019}.} also shows the deltas at $t=0$ computed by our proposed method, along with the ones computed using the DRL method. The DRL method produces percent errors of the deltas ranged from $1.01\%$ to as high as $37.71\%$; this is due to the limited sampling available under the run time constraint. The limited sampling leads to a poorly trained network and a degradation of the exercise boundary as shown in Figure \ref{fig:geo_exc}. We see better performance from our proposed model under our time constraints, as we consistently see errors below $2.519\%$. Furthermore, we compare the exercise boundaries computed by our approach with the ones computed using the DRL method.

\begin{table}[!htbp]
    \centering
    \caption{F1-score calculated using \eqref{eq:f1} for multi-dimensional American style geometric average call option under time constraint.}
    \begin{tabular}{ccccc|cccc}
        & \multicolumn{4}{c|}{Proposed method} & \multicolumn{4}{c}{DRL method} \\ \hline \hline
        $s_0$ & $d=2$ & $d=5$ & $d=30$ & $d=100$ & $d=2$ & $d=5$ & $d=30$ & d=100 \\ \hline
        90 & 0.9465 & 0.9486 & 0.9394 & 0.9545 & 0.6563 & 0.6550 & 0.6504 & 0.6565\\
        100 & 0.9241 & 0.9193 & 0.9390 & 0.9213 & 0.6644 & 0.6651 & 0.6453 & 0.6550\\
        110 & 0.8945 & 0.8953 & 0.9004 & 0.8956 & 0.6248 & 0.6403 & 0.6324 & 0.6193\\
    \end{tabular}
    \label{tab:geo_f1}
\end{table}

Table \ref{tab:geo_f1} evaluates the f1-score of the exercise boundary classification, as defined in \eqref{eq:f1}. For the proposed method, the f1-score ranges from $0.8945 - 0.9545$. For the DRL method, the f1-score ranges from $0.6193 - 0.6651$. This illustrates the more consistent performance of our proposed approach under time constraints. 

\begin{figure}[!htbp]
    \centering
    \begin{subfigure}{0.48\textwidth}
        \centering
        \includegraphics[scale=0.4]{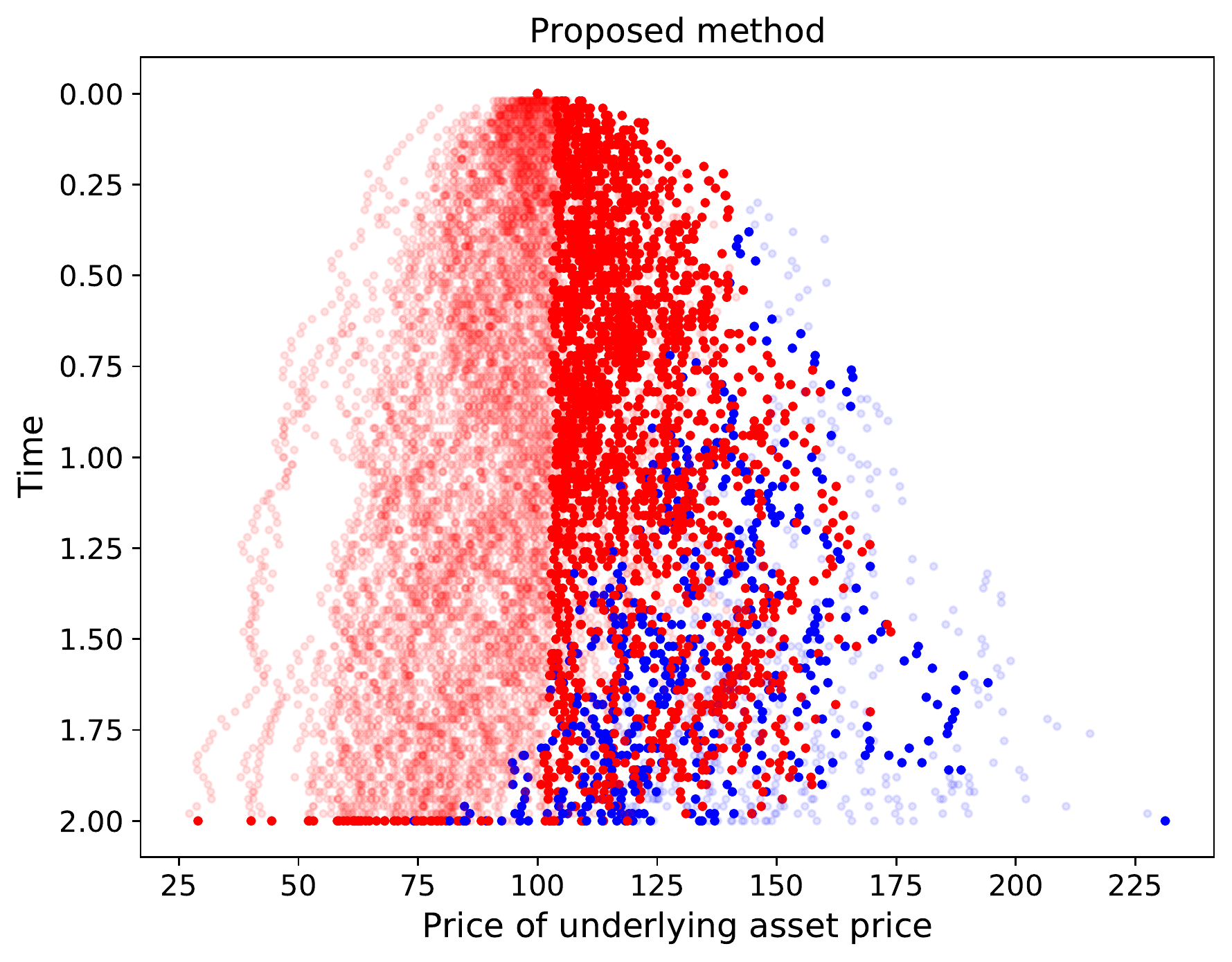}
        \caption{Proposed method: $d=5$}
    \end{subfigure}
    \begin{subfigure}{0.48\textwidth}
        \centering
        \includegraphics[scale=0.4]{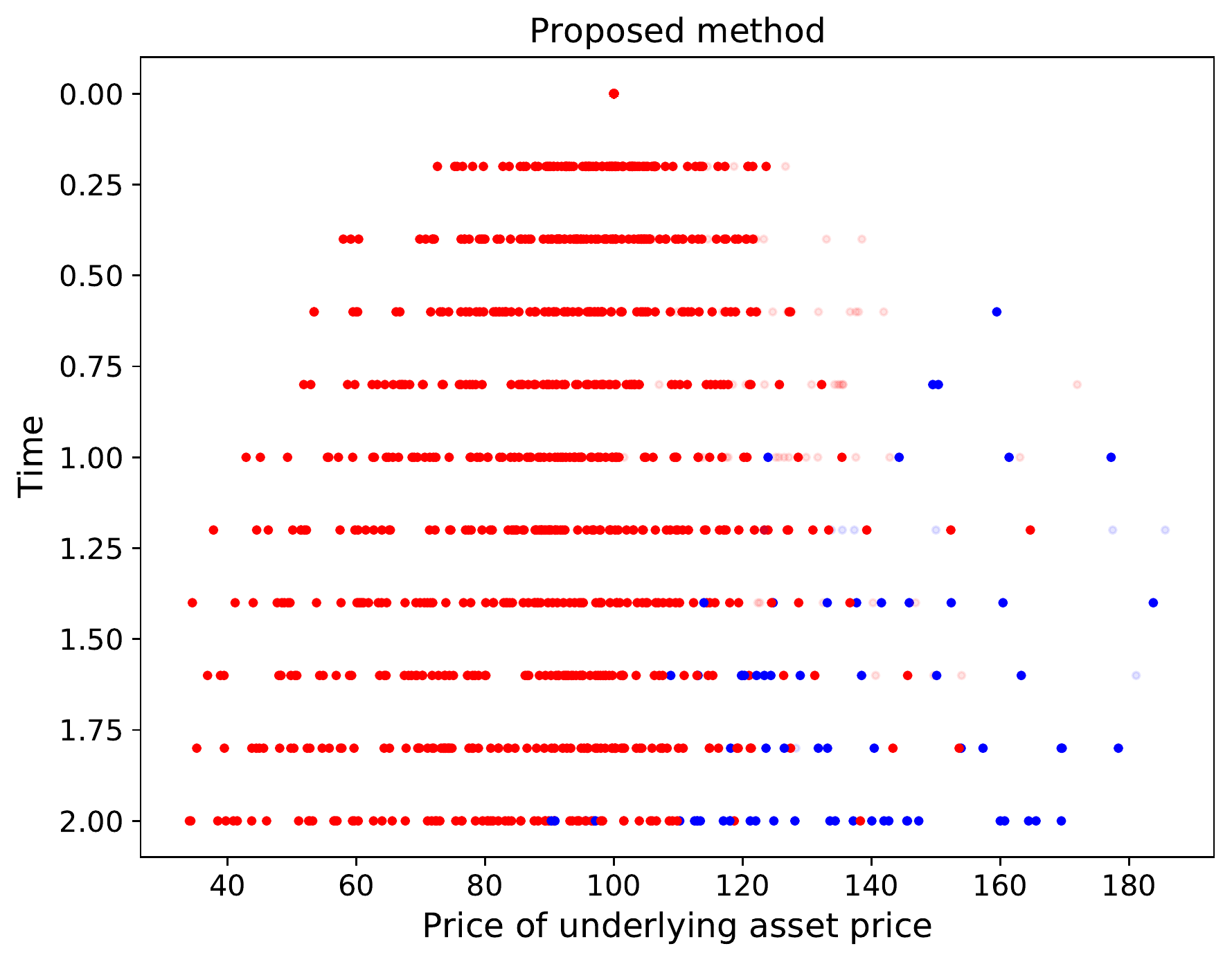}
        \caption{DRL method: $d=5$}
    \end{subfigure}
    \caption{Multi-dimensional geometric call options:  Comparison of exercise boundaries between the proposed neural network
    approach (left) and the method of DRL method (right). All blue points: sample points that
    should be exercised; all red points: sample points that should be continued; bold dark blue points: sample points that should be exercised
    but are misclassified as continued; bold dark red points: sample points that should be continued but are misclassified as exercised.}
    \label{fig:geo_exc}
\end{figure}

Figure \ref{fig:geo_exc} visualizes the exercise boundaries computed by
both methods. The results of our computation shows more points because the efficiency of our method allows us to use more timesteps for our simulation. We generate sample points on the entire spacetime, i.e. $\{(S_n^m, t_n)| n = 0,\dots,N; m = 1, \dots , M\}$; we classify each sample point using either our proposed method, or the DRL method. We use bold dark blue crosses to mark the sample points that should be exercised but are misclassified as continued (false positives), and bold dark red crosses to mark the ones that should be continued but are misclassified as exercised (false negatives). The plots show that under runtime constraints, the exercise boundary of our proposed method are more precise than the DRL method.

\subsubsection*{Experiment 2: Delta hedging of American style geometric average call option}

The purpose of experiment 2 is to perform delta hedging simulations over the period $[0, T]$ using our proposed method for the geometric call option. We evaluate the quality using the distribution of the relative profit and loss (P\&L) as described in \cite{forsyth-vetzal-2002} and \cite{he-2006}. Relative P\&L is given by
\begin{equation*}
    \textnormal{Relative}\ P\&L\equiv \frac{e^{-rT}\Pi_T}{V_0},
\end{equation*}
where $\Pi_T$ is the replicating portfolio at expiry $T$. For perfect hedging, the relative P\&L should be a Dirac delta function. 

Due to the discretization of time, the relative P\&L in \cite{chen-2019} is a normal distribution with mean zero and small standard deviation proportional to $\Delta t$ \cite{chen-2019}. We did not compare the hedging results with the DRL method in this experiment as the hedging intervals were different. We
emphasize that the computation of the relative P\&L must use both prices and deltas for the entire spacetime.

\begin{table}[!htbp]
    \centering
    \caption{Multi-dimensional geometric average call options: Computed means and standard deviations of the relative P\&Ls, subject to 250 hedging intervals.}
    \begin{tabular}{ccccc}
        & \multicolumn{4}{c}{Portfolio mean}\\ \hline \hline
        $s_0$ & $d=2$ & $d=5$ & $d=30$ & $d=100$ \\
        90 & -8.928e-8 & -1.064e-8 & -2.544e-9 & 7.430e-11 \\
        100 & -9.805e-8 & -5.697e-9 & -1.347e-8 & -1.931e-9 \\
        110 & -1.774e-7 & -4.014e-8 & -1.426e-8 & -4.312e-10 \\ \hline \\
    \end{tabular}
    \begin{tabular}{ccccc}
        & \multicolumn{4}{c}{Portfolio std dev} \\ \hline \hline
        $s_0$ & $d=2$ & $d=5$ & $d=30$ & $d=100$ \\
        90 & 1.718e-4 & 5.711e-5 & 1.107e-5 & 3.329e-6\\
        100 & 2.789e-4 & 9.691e-5 & 1.809e-5 & 5.446e-6\\
        110 & 3.9398e-4 & 1.458e-4 & 2.656e-5 & 8.035e-6\\ \hline
    \end{tabular}
    \label{tab:geo_portfolio}
\end{table}

Table \ref{tab:geo_portfolio} shows the means and the standard deviations of the
relative P\&Ls for all the $200000$ simulation paths, computed by our proposed method. The reported values are indeed close to zero and the standard deviations are also close to zero.

\begin{figure}[htbp]
    \centering
    \begin{subfigure}{0.3\textwidth}
        \centering
        \includegraphics[width=\textwidth]{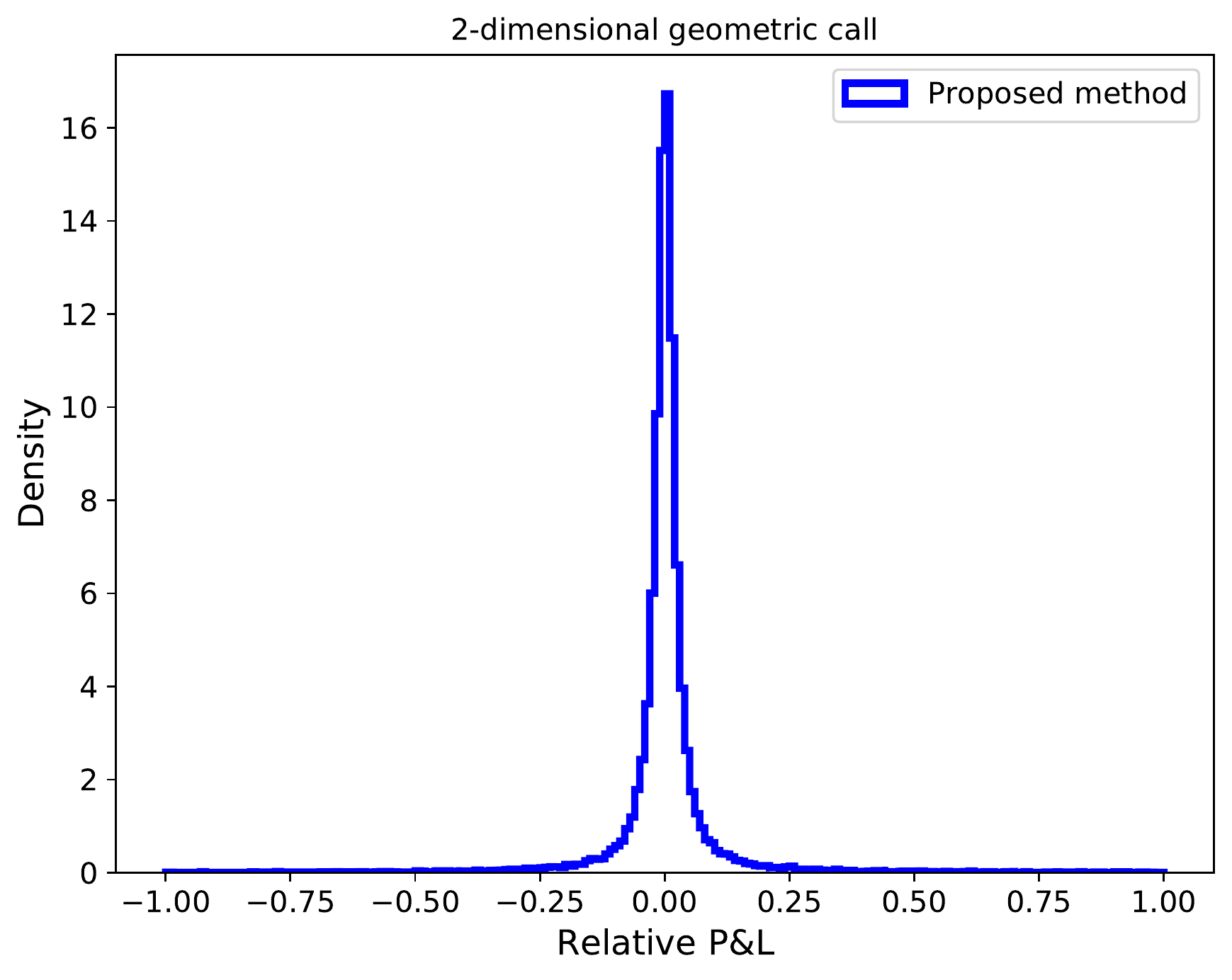}
        \caption{$d=2$}
    \end{subfigure}
    \hfill
    \begin{subfigure}{0.3\textwidth}
        \centering
        \includegraphics[width=\textwidth]{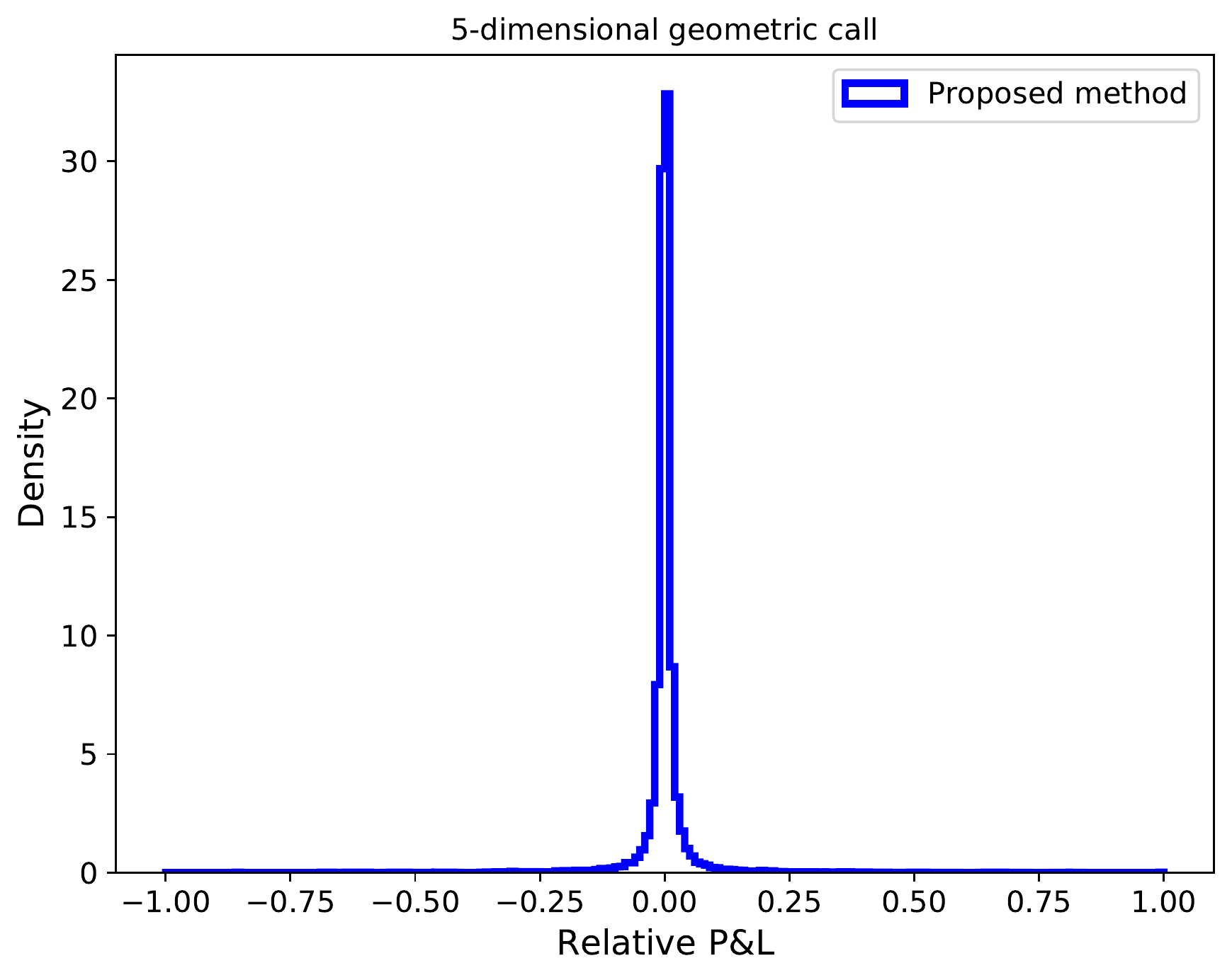}
        \caption{$d=5$}
    \end{subfigure}
    \hfill
    \begin{subfigure}{0.3\textwidth}
        \centering
        \includegraphics[width=\textwidth]{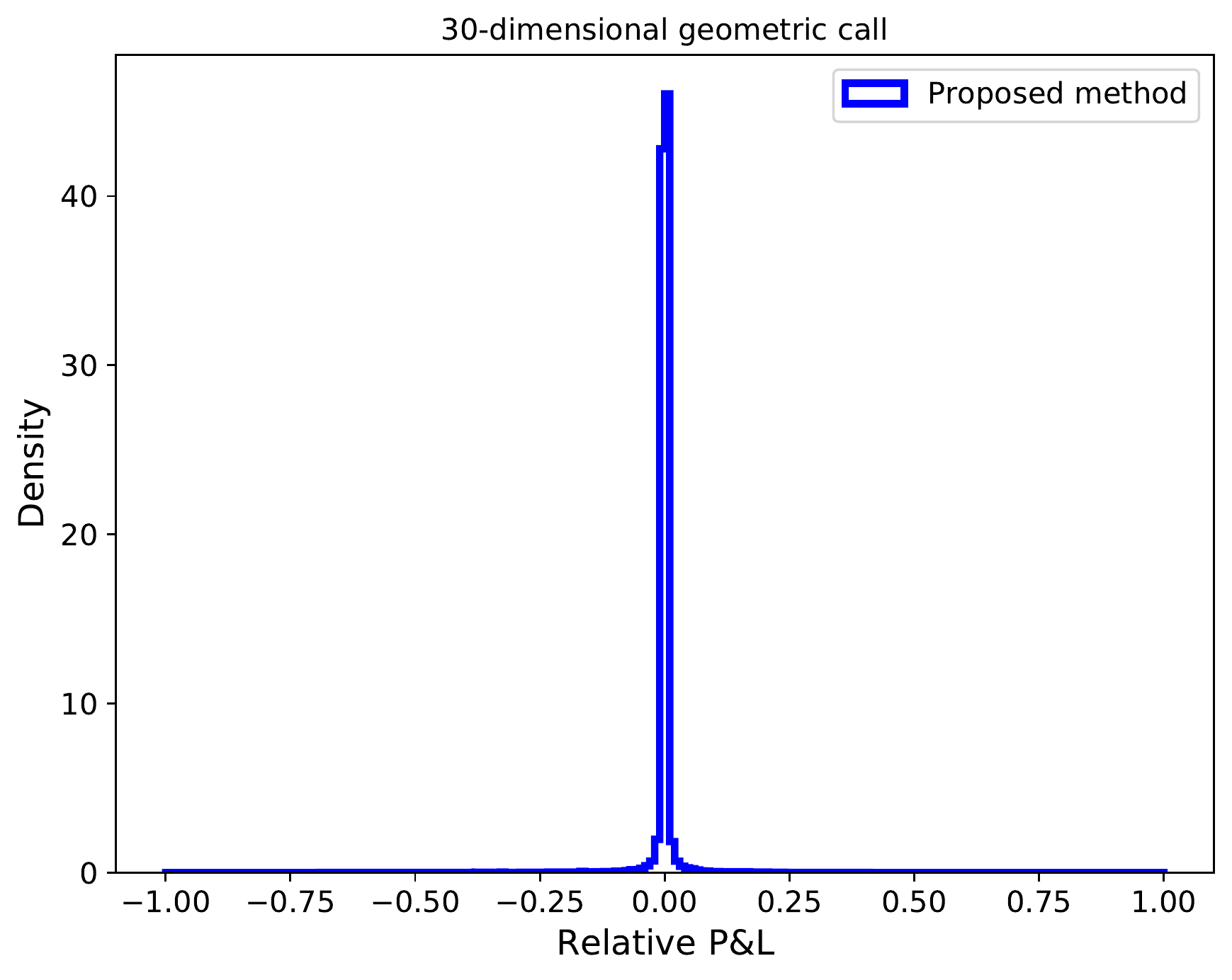}
        \caption{$d=30$}
    \end{subfigure}
    \caption{Multi-dimensional geometric call options: Distributions of the relative P\&Ls computed by the proposed neural network approach, subject to 250 hedging intervals. $s_0 = 100$}
    \label{fig:geo_portfolio}
\end{figure}

Figure \ref{fig:geo_portfolio} illustrates the distributions of the relative P\&Ls. The resulting distributions are near perfect, and approaches the Dirac function with zero means. This further illustrates the capabilities of our approach in determining prices and deltas across spacetime These results confirm that our methods provide a good hedge of assets across the spacetime prices and the spacetime deltas computed by the proposed method.

\subsection{American style Max Call Option}

We consider a $d$-dimensional American style ‘max’ call
option, where $\rho_{ij} = \rho$ for $i\neq j$, $\sigma_i = \sigma$ for all $i$’s, and the
payoff function is given by $f(\vv{s}) = \max\{\max\limits_{i=1,\dots,d}^{}\{s_i\} - K, 0\}$. Multi-dimensional max options are common in practical applications.

In the following Experiments $3$, $4$ and $5$, we evaluate the American style max call option. Like Experiments $1$ and $2$, the methods presented in Experiment $3$ is run such that the runtime between methods match. This is done to show that our method performs more efficiently with other payoff structures. In this experiment, we also compare our method to the unsupervised method and investigate the performance of our method far ITM and far OTM, i.e. $s_0=80,120$. Experiment $4$ and $5$ are run until the accuracy of our proposed method matches the accuracy of the DRL method. In Experiment $5$, we also look at the timing results of the Longstaff-Schwartz method that matches the accuracy of the DRL method. For $d=2$, we computed price of the American style max call option as described in \cite{broadie-1996}, we set $\{T=1, r=0.05, \sigma = 0.2, \rho = 0.3\}$. For $d=5,30$ we computed the price of the American max call option as described in \cite{firth-2005}. We set $\{T=3, r=0.05, \sigma = 0.2, \rho = 0\}$.

\subsubsection*{Experiment 3: 2-dimensional American style max call option} \label{ex:3}

The purpose of this example is to show that our method can be used to get the price and delta of different payoff functions. More specifically, we compare our method with the DRL method, using the the finite difference method solution as a ground truth. We approximate the exact prices and deltas by the Crank-Nicolson finite difference method with $1000$ timesteps and $2049 \times 2049$ space grid points as in \cite{chen-2019}. The DRL method for this experiment was run with $J=8$, with sample size $200000$, for $N=10$ timesteps. Our experiments compare the runtime of our proposed method and the DRL method given the same runtime. For $s_0 = 90,100,110$, we do not compare our method to the Longstaff-Schwartz method as this comparison for these $s_0$, since that has been done in \cite{chen-2019}. However we present additional results for $s_0=80,120$, where we make the comparison to the Longstaff-Schwartz method with a $4^\textnormal{th}$ degree polynomial.\vskip 2mm
\noindent\textit{Comparison with DRL Method}\vskip 2mm
Using our proposed method, the percent errors of the computed prices at $t = 0$ are less than $2.062\%$ (Table \ref{tab:2-d_max}), the percent errors of the computed deltas at $t = 0$ are less than $2.3665\%$ (Table \ref{tab:2-d_max}). For the same runtime, our method produces smaller errors than the DRL method, which had errors between $4.2139\%-27.6173\%$ for price and $4.3196\%-66.1151\%$ for delta.
\begin{table}[!htbp]
    \centering
    \caption{2-dimensional Max call: computed option prices and deltas at $t=0$.}
    \begin{tabular}{ccccccc}
        & \multicolumn{6}{c}{Finite Difference method} \\ \toprule
        $s_0$ & Price & Std & \% Error & Delta & \% Error & Run Time (s) \\ \colrule
        90 & 4.1941 & -- & -- & [20.2762, 20.2353] & -- & --\\
        100 & 9.6309 & -- & -- & [32.8168, 32.8905] & -- & --\\
        110 & 17.3235 & -- & -- & [42.4883, 42.4207] & -- & --\\ \botrule \\
        & \multicolumn{6}{c}{Proposed method} \\ \toprule
        $s_0$ & Price & Std & \% Error & Delta &  \% Error & Run Time (s) \\ \colrule
        90 & 4.2806 & 0.00162 & 2.062421 & [20.7641, 20.7061] & 2.3665 & 150.3\\
        100 & 9.7198 & 0.00177 & 0.923071 & [33.3599, 33.3998] & 1.6016 & 156.39\\
        110 & 17.4483 & 0.00168 & 0.72041 & [43.3701, 43.1945] & 1.9499 & 150.61\\ \botrule \\
        & \multicolumn{6}{c}{DRL method} \\ \toprule
        $s_0$ & Price & Std & \% Error & Delta &  \% Error & Run Time (s) \\ \colrule
        90 & 3.0358 & 0.3946 &  27.61737 & [27.4312, 27.3825] & 35.3041 & 150.4\\
        100 & 8.0745 & 0.3856 & 16.16048 & [39.4164, 39.3535] & 66.1151 & 156.43\\
        110 & 16.5994 & 0.3761 &  4.2139 & [40.5587, 40.6825] & 4.3197 & 150.85\\ \botrule
    \end{tabular}
    \label{tab:2-d_max}
\end{table}
\begin{figure}[!htbp]
    \centering
    \begin{subfigure}{0.48\textwidth}
        \centering
        \includegraphics[scale=0.4]{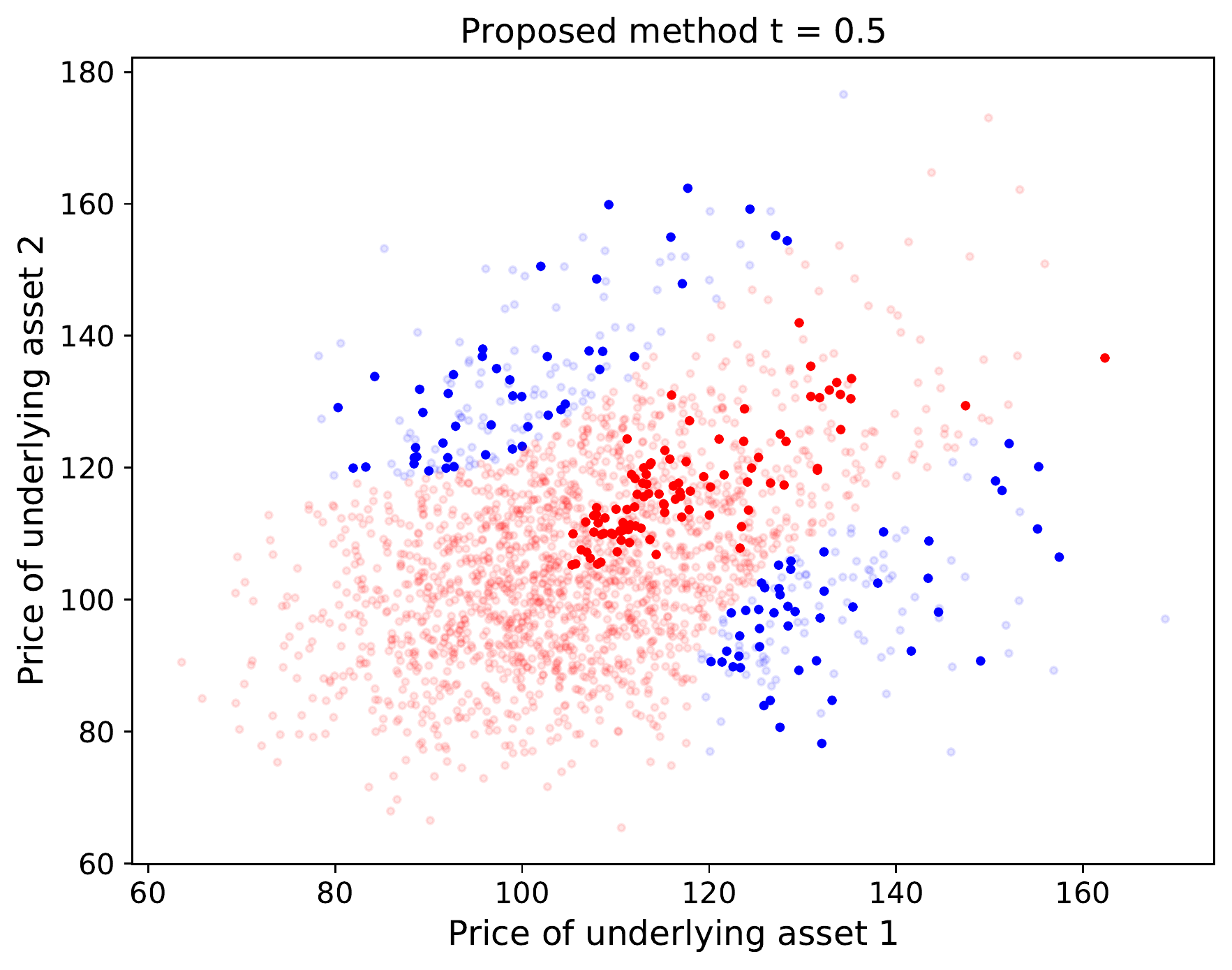}
        \caption{Proposed method:  $t=0.5$}
    \end{subfigure}
    \begin{subfigure}{0.48\textwidth}
        \centering
        \includegraphics[scale=0.4]{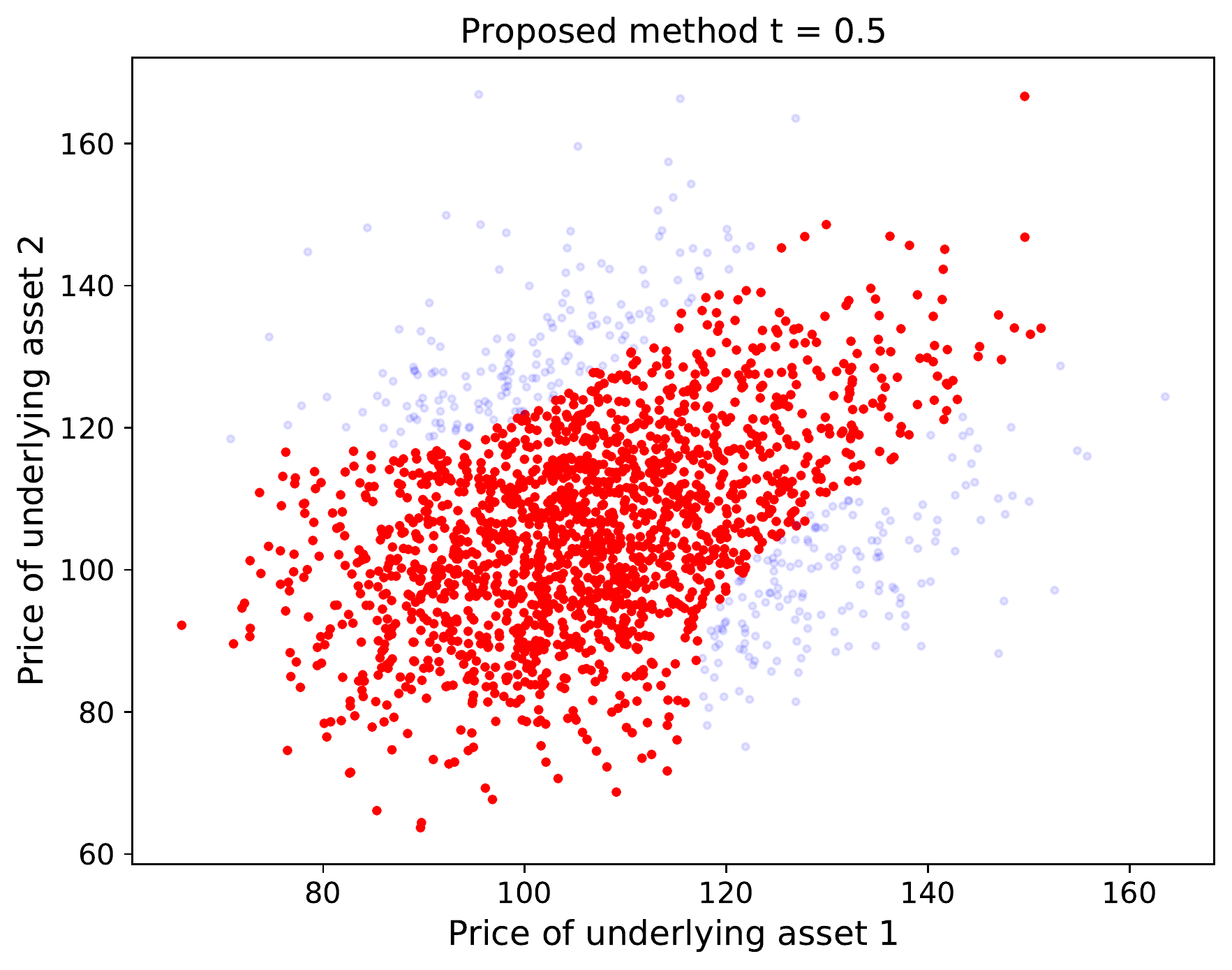}
        \caption{DRL method: $t=0.5$}
    \end{subfigure}
    \caption{2-dimensional max call options:  Comparison of exercise boundaries between the proposed RNN
    approach (top left and bottom left) and the method of DRL method (top right and bottom right) at $t=0.5$ (top), $t=0.75$ (bottom) and $s_0=110$. All blue points: sample points that
    should be exercised; all red points: sample points that should be continued; bold dark blue points: sample points that should be exercised
    but are misclassified as continued; bold dark red points: sample points that should be continued but are misclassified as exercised.}
    \label{fig:max_exc}
\end{figure}
\vskip 2mm
\noindent\textit{Delta hedging of American style max call option}\vskip 2mm
In addition, we compute the relative P\&Ls of the delta hedging portfolio by the finite difference method and compare them with the values computed by our approach. Table \ref{tab:max_f1} shows the F1 score of the two methods. The F1 score of our method ranges from $0.8011 - 0.9564$. The F1 score of the DRL method ranges from $0.6667 - 0.7965$.
\begin{table}[htbp]
    \begin{center}
    \caption{2-dimensional Max call option: Means and standard deviations of the relative P\&Ls by finite difference versus the proposed method, subject to 100 hedging intervals}
    {
        \begin{tabular}{ccc|cc}
        \multicolumn{5}{c}{ (a) 100 hedging intervals}\\
        & \multicolumn{2}{c|}{Proposed method} & \multicolumn{2}{c}{Finite difference method} \\ \hline \hline
        $s_0$ & Mean & Std & Mean & Std \\ \hline
        90 & -0.0018 & 0.0899 & 0.022 & 0.1932 \\
        100 & -0.0022 & 0.0434 & 0.0016 & 0.0990 \\
        110 & -0.0023 & 0.0072 & 0.0016 & 0.0614 \\ \hline
        \end{tabular}
    \vfill
        \begin{tabular}{ccc|cc}
        \multicolumn{5}{c}{(b) 500 hedging intervals}\\
        & \multicolumn{2}{c|}{Proposed method} & \multicolumn{2}{c}{Finite difference method} \\ \hline \hline
        $s_0$ & Mean & Std & Mean & Std  \\ \hline
        90 & -0.0018 & 0.0011 & 0.022 & 0.1932 \\
        100 & -0.0028 & 0.0017 & 0.0016 & 0.0990 \\
        110 & -0.0029 & 0.0021 & 0.0016 & 0.0614 \\ \hline
        \end{tabular}
    \label{tab:max_portfolio}
    }
    \end{center}
\end{table}
\begin{figure}[!htbp]
    \centering
    \begin{subfigure}{0.48\textwidth}
        \centering
        \includegraphics[scale=0.4]{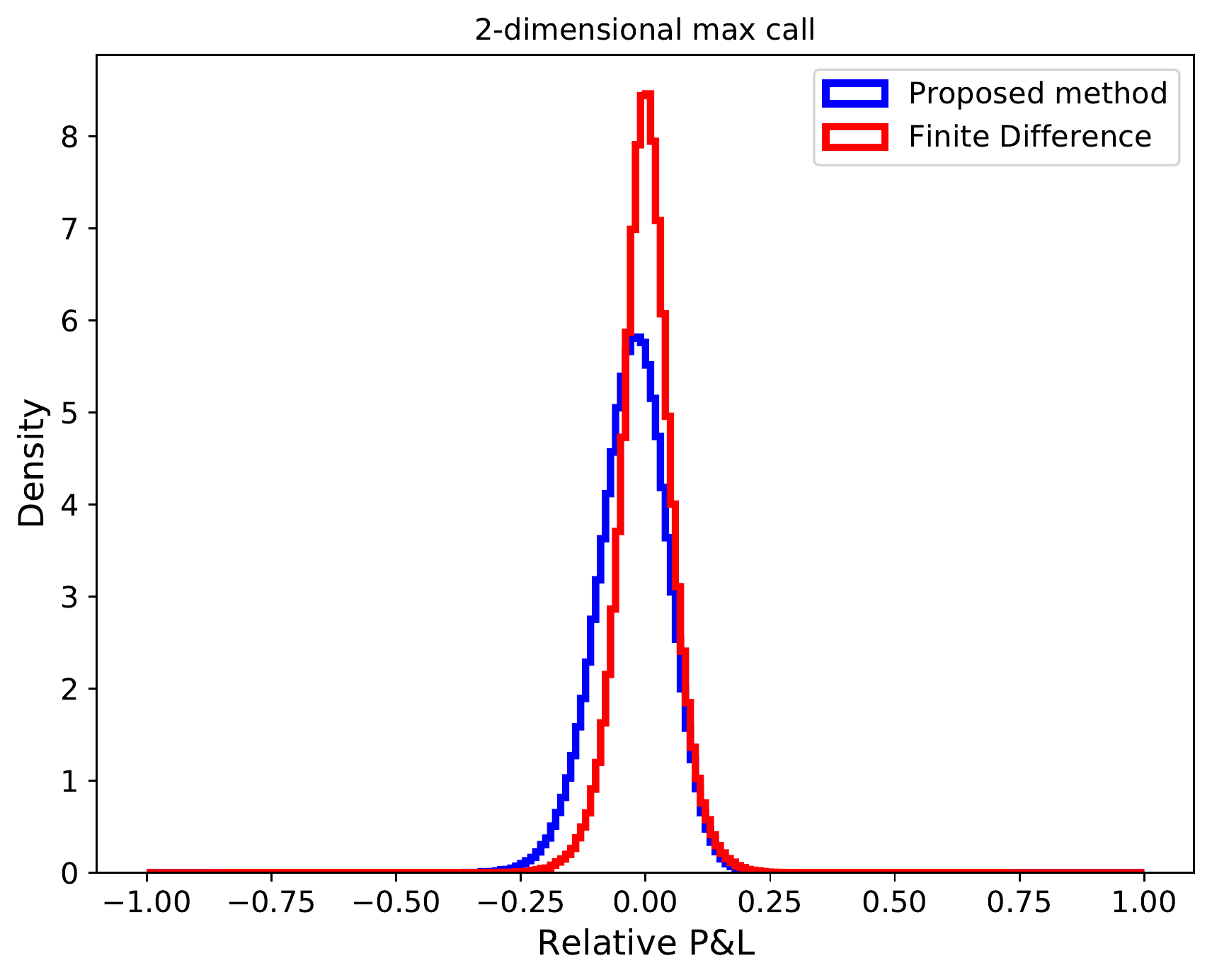}
        \caption{100 hedging intervals}
    \end{subfigure}
    \begin{subfigure}{0.48\textwidth}
        \centering
        \includegraphics[scale=0.4]{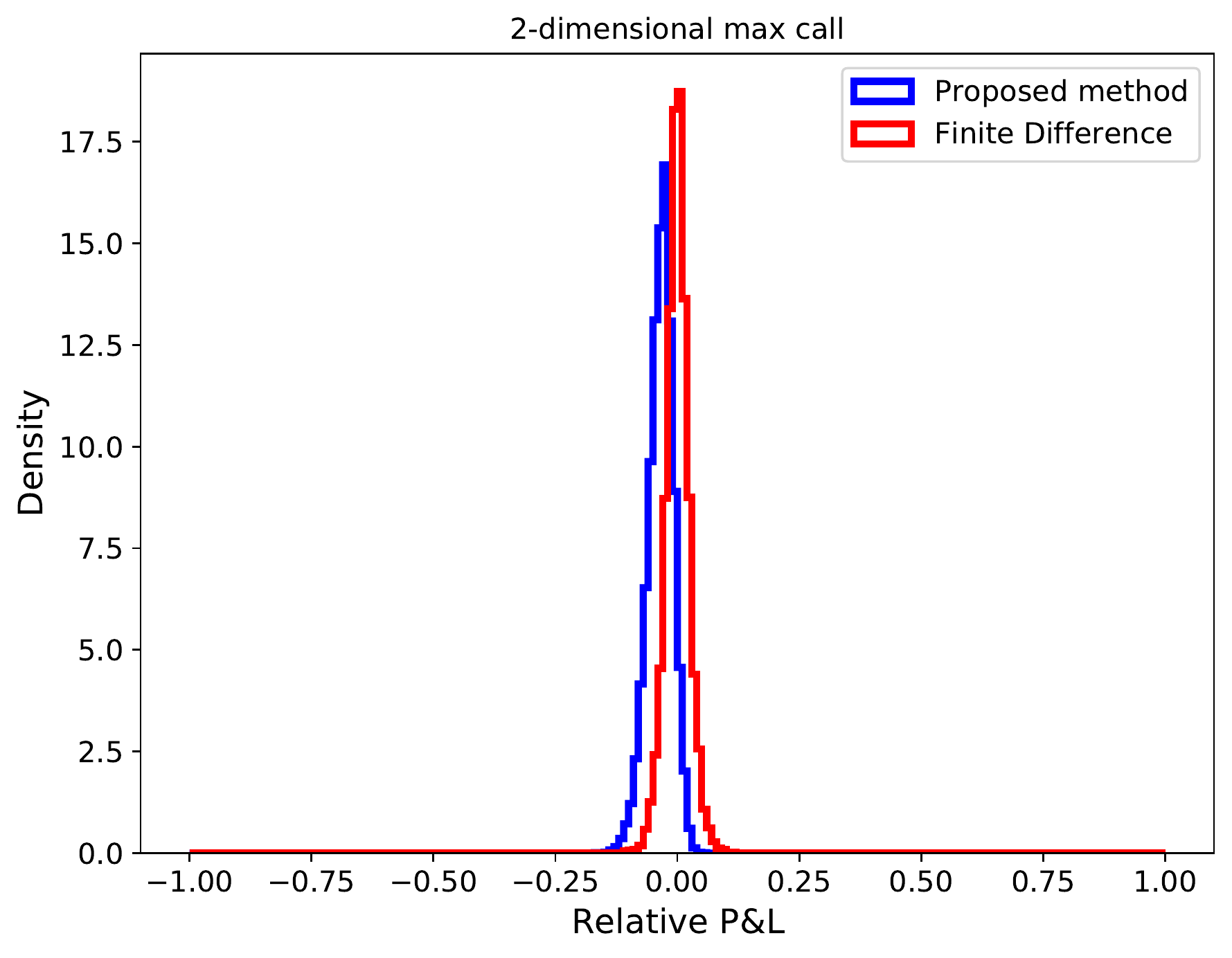}
        \caption{500 hedging intervals}
    \end{subfigure}
    \caption{2-dimensional Max call option: Comparison of the distributions of the relative P\&Ls computed by the proposed RNN approach (blue) versus by finite difference method (red) at $s_0=1.1K$.}
    \label{fig:max_portfolio}
\end{figure}

Table \ref{tab:max_portfolio} and Figure \ref{fig:max_portfolio} show the means, standard deviations and the distributions of the relative P\&Ls computed by the proposed approach versus by finite difference methods. The variation computed by our proposed approach decreases as more timesteps are used in the training of the model, however, we can observe that the proposed method produces a skewed normal P\&L distribution. The right skewness of the resulting P\&L distribution shows that the replicating portfolio of our methodology is biased to over estimating the risk of the portfolio.
\vskip 2mm
\noindent\textit{Far OTM and ITM performance}\vskip 2mm
Next we looked at testing the performance of our proposed method in the case where the max call option was far ITM, $s_0=80$, and far OTM, $s_0=120$. We compared this to the DRL method and the Longstaff-Schwartz method and summarize the results in Table \ref{tab:2-d_max-far}.
\begin{table}[!htbp]
    \centering
    \caption{F1-score calculated using \eqref{eq:f1} for 2-dimensional American style max call option under time constraint.}
    \begin{tabular}{ccc|cc}
        & \multicolumn{2}{c|}{Proposed method} & \multicolumn{2}{c}{DRL method} \\ \hline \hline
        $s_0$ & $t=0.5$ & $t=0.75$ & $t=0.5$ & $t=0.75$ \\ \hline
        90 & 0.9564 & 0.9405 & 0.7187 & 0.7163 \\
        100 & 0.8983 & 0.8780 & 0.7965 & 0.7471 \\
        110 & 0.8362 & 0.8011 & 0.6667 & 0.6667 \\ \hline
    \end{tabular}
    \label{tab:max_f1}
\end{table}

Using our proposed method, the percent errors of the computed prices at $t = 0$ are less than $2.7671\%$ (Table \ref{tab:2-d_max}), the percent errors of the computed deltas at $t = 0$ are less than $3.8461\%$ (Table \ref{tab:2-d_max-far}). For the same runtime, our method produces smaller errors than the DRL method, which had errors $13.5132\%$ and $43.3926\%$ for prices and $44.32096\%$ and $6.3112\%$ for deltas. However this is not the case for the Longstaff-Schwartz method. We see that for small dimensions the price approximated by the Longstaff-Schwartz method has errors of $1.6020\%$ and $0.6170\%$, which outperforms both our proposed method and the DRL method under time constraint. However, our proposed method still outperforms the Longstaff-Schwartz method when it comes to computing deltas. The Longstaff-Schwartz method had relative errors of $8.6042\%$ and $3.1955\%$.
\begin{table}[!htbp]
    \centering
    \caption{2-dimensional Max call: Far OTM and ITM computed option prices and deltas at $t=0$.}
    \begin{tabular}{ccccccc}
        & \multicolumn{6}{c}{Finite Difference method} \\ \hline \hline
        $s_0$ & Price & Std & \% Error & Delta &  \% Error & Run Time (s) \\ \hline
        80 & 1.3046 & -- & -- & [8.9656, 8.9176] & -- & --\\
        120 & 26.5793 & -- & -- & [47.9348, 47.7876] & -- & --\\ \hline \\
        & \multicolumn{6}{c}{Proposed method} \\ \hline \hline
        $s_0$ & Price & Std & \% Error & Delta &  \% Error & Run Time (s) \\ \hline
        80 & 1.3407 & 0.00163 & 2.7671 & [9.3101, 9,2609] & 3.8461 & 151.82\\
        120 & 26.826 & 0.00167 &  0.9281 & [49.3198, 49.1329] & 2.8523 & 153.67\\ \hline \\
        & \multicolumn{6}{c}{DRL method} \\ \hline \hline
        $s_0$ & Price & Std & \% Error & Delta &  \% Error & Run Time (s) \\ \hline
        80 & 0.7385 & 0.3891 & 43.3926 & [12.9289, 12.8604] & 44.2096 & 151.78\\
        120 & 22.9876 & 0.3742 & 13.5132 & [50.8435, 50.9202] & 6.3112 & 153.56\\ \hline \\
        & \multicolumn{6}{c}{Lonstaff-Schwartz} \\ \hline \hline
        $s_0$ & Price & Std & \% Error & Delta &  \% Error & Run Time (s) \\ \hline 
        80 & 1.2837 & 0.01369 & 1.6020  & [9.5721, 9.8498] & 8.60416 & 151.14\\
        120 & 26.4153 & 0.01334 &  0.6170 & [49.0682, 49.713] & 3.19549 & 153.13\\ \hline
    \end{tabular}
    \label{tab:2-d_max-far}
\end{table}

For pricing under limited runtime, Longstaff-Schwartz method performs better than both neural network approaches when $d=2$, however, we see that the delta performance is not as good as our proposed method.
\vskip 2mm
\noindent\textit{Price comparison with unsupervised method}\vskip 2mm
We also compare the output of the proposed method to the unsupervised method using the L-BFGS optimizer and the loss function outlined in \cite{salvador-2020}. We implemented the unsupervised method as in \cite{salvador-2020}, where an $101\times 101 \times 75$ spacetime grid was used to run the unsupervised method. We let the method run until the running time is close to the times given in Table \ref{tab:2-d_max} and we found the prices for initial price of $s_0 = 90,100,110$; the resulting value output by the unsupervised method, $V_{unsupervised} = 25.72,25.00,24.59$, respectively.

The pricing errors are computed based on FDM prices in Table \ref{tab:2-d_max} for each $s_0$. The relative error of the unsupervised method for each $s_0$ are $[513.481\%, 167.160\%, 48.527\%]$, respectively. We found that due to the limited training time the unsupervised method could not be trained until convergence. This led to drastically inaccurate prices. The unsupervised method is also not designed to compute deltas. The unsupervised method can be classified as a deep Galerkin type method \cite{sirignano-2018} which learns the pricing function across the given grid domain. The deep Galerkin type method has been compared to BSDE methods in \cite{chen-2019} and it was found that the delta errors computed from the autodifferentiation of the neural network pricing functions over the grid is relatively high compared to BSDE and HJB methods.

\subsubsection*{Experiment 4: 5-dimensional American style max call option} \label{ex:4}

We note that unlike Experiment $3$, here the exact solutions are not available. Therefore, in this section we run our proposed method until it reaches a level of accuracy similar to the results produced by the DRL method. We do not compare with the Longstaff-Schwartz method as this has been done in \cite{chen-2019} for this experiment. The DRL method for this experiment was run with $J=2$, with sample size $720000$, for $N=100$ timesteps.

Table \ref{tab:5-d_max} reports the runtime, option prices and deltas at $t = 0$ computed by the proposed method and the DRL method. This experiment shows that for our method to achieve prices that have an absolute difference of order $10^{-2}$ to the price results from the DRL method following the setup in \cite{chen-2019}. For $s_0=90,100,110$ the DRL method requires $10026.46$ seconds, $10592.07$ seconds and $10171.40$ respectively. Our method requires $939.86$ seconds, $953.41$ seconds and $944.85$ seconds respectively. Our method shows about a $10$ times faster run time to achieve similar accuracies. 

\begin{table}[!htbp]
    \centering
    \caption{5-dimensional max call option: computed prices and deltas at $t=0$}
    \begin{tabular}{ccccc}
        & \multicolumn{4}{c}{Proposed method} \\
        \hline \hline
        $s_0$ & Price & Std & Delta & Run Time (s) \\ \hline
        90 & 16.9152 & 0.000851 &[18.196,18.060,18.142,18.173,18.206] & 939.86\\
        100 & 26.5210 & 0.000999 &[19.100,19.117,19.087,19.511,19.072] & 953.41\\
        110 & 37.1389 & 0.000813 &[20.475,20.646,20.678,20.443,20.601] & 944.85 \\ \hline \\
        & \multicolumn{4}{c}{DRL method} \\
        \hline \hline
        $s_0$ & Price & Std &Delta & Run Time (s) \\ \hline
        90 & 16.8896 & 0.000709 &[17.280, 17.320, 17.540, 17.380, 17.470] & 10026.46\\
        100 & 26.4876 & 0.000728 &[20.170, 20.040, 19.980, 20.710, 20.410] & 10592.07\\
        110 & 37.0996 & 0.000721 &[21.570, 21.980, 21.900, 21.490, 22.020] & 10171.40 \\ \hline
    \end{tabular}
    \label{tab:5-d_max}
\end{table}

\subsubsection*{Experiment 5: 30-dimensional American style max call option} \label{ex:5}
We extend Experiment $4$ to 30-dimensions, with a simulation size of $270000$. We report the runtime, mean option price, option price standard deviation, mean delta and delta standard deviation that the methods took to achieve prices within the order of $10^{-2}$ absolute difference from the DRL method. In this example we also make comparison with the Longstaff-Schwartz method, in addition we also report standard deviation for price and delta, and report mean price and delta. The experimental results are summarized in Table \ref{tab:30-d_max}.
\begin{table}[!!htbp]
    \centering
    \caption{30-dimensional max call option: computed prices and deltas at $t=0$}
    \begin{tabular}{cccccc}
        & \multicolumn{5}{c}{Proposed method} \\
        \hline \hline
        $s_0$ & Price & Std & Delta & Std & Run Time (s) \\ \hline
        90 & 41.6933 & 0.00175 & 4.9568 & 0.0230 & 1384.35\\
        100 & 57.382 & 0.00195 & 4.9354 & 0.021 & 1388.06\\
        110 & 72.3044 & 0.00146 & 4.9315 & 0.0232 & 1389.65 \\ \hline \\
        & \multicolumn{5}{c}{DRL method} \\
        \hline \hline
        $s_0$ & Price & Std & Delta & Std & Run Time (s) \\ \hline
        90 & 41.7125 & 0.00152 & 5.0139 & 0.0117 & 15763.37\\
        100 & 57.4205 & 0.00148 & 5.0121 & 0.0108 & 15775.14\\
        110 & 72.3504 & 0.00186 & 5.0133 & 0.0137 & 15743.56 \\ \hline \\
        & \multicolumn{5}{c}{Longstaff-Schwartz} \\
        \hline \hline
        $s_0$ & Price & Std & Delta & Std & Run Time (s) \\ \hline
        90 & 41.6844 & 0.00259 & 4.8842 & 0.03386 & 47034.76\\
        100 & 57.3705 & 0.00283 & 4.8221 & 0.03459  & 47094.21\\
        110 & 72.5003 & 0.00309 & 4.9585 & 0.03614 & 47099.81 \\ \hline
    \end{tabular}
    \label{tab:30-d_max}
\end{table}

Table \ref{tab:30-d_max} reports the runtime, option prices and deltas at $t = 0$ computed by the proposed method, the DRL method and the Longstaff-Schwartz method. This experiment shows that for our method to achieve prices that have an absolute difference of order $10^{-2}$ to the price results from the DRL method following the setup in \cite{chen-2019}. For $s_0=90,100,110$ the DRL method requires $15763.37$ seconds, $15783.14$ seconds and $15743.56$ seconds respectively. Our method required $1384.35$ seconds, $1388.06$ seconds and $1389.65$ seconds respectively. The Longstaff-Schwartz method required $47034.76$ seconds, $47094.21$ seconds and $47099.81$ seconds respectively. Our method shows about a $10$ times faster run time to achieve similar accuracy over the DRL method and over $35$ times faster than the Longstaff-Schwartz method. We can also see that the delta of the option using the pathwise derivative for the Longstaff-Schwartz method has a larger variation over the deltas found using the DRL method and the Proposed method.

\section{Conclusion}\label{sec:Conclusion}

We propose a deep Recurrent network framework for pricing and hedging high-dimensional American option. Our proposed method achieves a better runtime efficiency than the DRL method. We solve the BSDE using a neural network parameterization, however, we use the loss function presented in \cite{hure-2020}, which has theoretical convergence results. The theoretical time complexity of our proposed method is both linear in the number of assets $d$ and the number of timesteps $N$ which is better than the quadratic results of \cite{chen-2019}. However, in our experimentation we see that this is only the case for $d > 20$ and at $N = 500$. This is largely attributed to the timeskip step implemented by \cite{chen-2019}, which effectively makes the time complexity of the DRL method slightly better than quadratic. Interestingly we found that as $N$ decreases there is an constant increase of the quadratic complexity term. Our Algorithms \ref{alg:training} and \ref{alg:eval} yields prices and deltas at all points in space and time and it is more accurate than the DRL method under time constraints. The main drawback of our method is that, when there are no time constraints, our method is not as accurate. This can be explained by the accuracy difference between different loss functions used to solve \eqref{eq:BSDE} as shown in \cite{hure-2020}. A potential future work is to improve the accuracy of the method by improving Recurrent network architecture, as well as extension to transformer networks. Another avenue of work can also explore the use of neural network sample generators to generate better sample data for training.

\bibliographystyle{unsrtnat}
\bibliography{References}

\end{document}